\documentclass[journal,12pt,onecolumn,draftclsnofoot]{IEEEtran}
\usepackage{eurosym}
\usepackage{cite,algorithm,array,url}
\usepackage{graphicx,colortbl}
\usepackage{amsfonts,times,amsmath,amssymb,amsthm}
\usepackage{multirow,booktabs}
\usepackage{color}
\usepackage{algorithm}
\usepackage{algpseudocode}

\ifCLASSOPTIONcompsoc
\usepackage[caption=false,font=normalsize,labelfon
t=sf,textfont=sf]{subfig}
\else
\usepackage[caption=false,font=footnotesize]{subfi
g}
\fi

\definecolor{darkpastelgreen}{rgb}{0.01, 0.75, 0.24}
\definecolor{orange}{rgb}{1.0, 0.55, 0.0}

\graphicspath{{images/}}

\theoremstyle{plain}
\newtheorem{theorem}{Theorem}
\newtheorem{lemma}{Lemma}

\theoremstyle{definition}

\theoremstyle{remark}
\newtheorem{remark}{Remark}

\begin{document}

\title{Sample Greedy Gossip for Distributed Network-Wide Average Computation}
\author{Hyo-Sang~Shin\textsuperscript{*}, Shaoming~He and Antonios~Tsourdos\thanks{
Hyo-Sang Shin, Shaoming He and Antonios Tsourdos are with the School of Aerospace, Transport and Manufacturing, Cranfield University, Cranfield MK43 0AL, UK
(email: $\left\{\text{h.shin,shaoming.he,a.tsourdos}\right\}$@cranfield.ac.uk)}
\thanks{\textsuperscript{*}Corresponding Author.}}
\maketitle

\begin{abstract}
This paper investigates the problem of distributed network-wide averaging and proposes a new  greedy gossip algorithm. Instead of finding the optimal path of each node in a greedy manner, the proposed approach utilises a suboptimal communication path by performing greedy selection among randomly selected active local nodes. Theoretical analysis on convergence speed is also performed to investigate the characteristics of the proposed algorithm. The main feature of the new algorithm is that it provides great flexibility and well balance between communication cost and convergence performance introduced by the stochastic sampling strategy. Extensive numerical simulations are performed to validate the analytic findings. 
\end{abstract}

\begin{IEEEkeywords}
Partially-connected network, Distributed averaging, Gossip process, Sample greedy
\end{IEEEkeywords}


\section{Introduction}
\label{sec:1}

Employment of multiple mobile robots to cooperatively perform complex tasks has become a viable option along with dramatic technical advancements in low-cost, lightweight and power efficient robots. Average computation or average agreement over a wireless agent network is recognised as a key enabler for the operation of such a multiple robot system \cite{talebi2016distributed,he2018multi,jia2016cooperative,peterson2015exploitation,rabbat2006decentralized,he2019distributeda,hlinka2013distributed,he2019distributed}.  To this end, this paper investigates the problem of network-wide averaging, especially based on the distributed architecture.

With the development of network theory, the control-theoretic consensus algorithm \cite{olfati2007consensus,carli2008distributed,ren2007information,ren2018distributed} has become a popular and powerful tool in network-wide average computation. At each iteration, each local agent broadcasts its information to locally-connected neighbours and leverages the received information for update. Although average consensus provides a general framework for network-wide computation, the information exchange typically requires each local agent to be accessible by all its local agents \cite{xiao2004fast,ustebay2010greedy}. In contrast, gossip message passing algorithms only utilise two local agents in information exchange at every iteration and thus shows practical attractiveness and robustness against topology variation \cite{fagnani2008randomized,sarwate2009impact,iutzeler2013analysis}. 

One pioneering work of gossip algorithms was the randomised gossip, proposed in \cite{boyd2006randomized}. At each communication round, one local node randomly wakes up and randomly picks one of its neighbouring node for information exchange. These two agents then update their information as the average of their prior information. Although the randomised gossip guarantees asymptotic convergence, the convergence speed is very slow due to the randomised nature. An improvement over randomised gossip was geographic gossip \cite{dimakis2008geographic}, which leverages the geographic routing to improve the convergence speed of randomised gossip. The basic idea behind geographic gossip is to utilise a multi-hop communication strategy with full knowledge of geographic property of the network. This communication structure, however, largely relies on the reliability of network communications over many hops and inevitably involves communication overhead due to geographic routing. Unlike geometric gossip, the broadcast gossip, proposed in\cite{aysal2009broadcast,aysal2008broadcast,fagnani2011broadcast,franceschelli2010distributed}, makes use of the broadcast nature of wireless sensor networks. In other words, a local node is activated uniformly and broadcasts its information to connected neighbours at each iteration. These neighbouring nodes then update their information by a weighted average of their own and the broadcasted information. Although the convergence of broadcast gossip is relatively fast, this algorithm cannot guarantee the convergence to the true average, i.e., it introduces bias. In a recent noteworthy contribution \cite{ustebay2010greedy}, the authors suggested a new variant of gossip process, termed as greedy gossip, for distributed averaging. Unlike randomised gossip, the greedy gossip finds the optimal communication path that provides the maximum information discrepancy for a random local agent. It has been theoretically proved that the greedy gossip can significantly improve the convergence speed at the price of higher communication burden, compared to randomised gossip.

The two types of gossip, randomised and greedy, algorithms exhibit complementary characteristics.  More specifically, the randomised gossip has lower computational burden at each gossip iteration, but its convergence rate is relatively slow due to the randomised nature. On the other hand, the greedy gossip approach enjoys faster convergence by leveraging a deterministic communication strategy, compared with the randomised gossip. However, it requires communications of each node with all its neighbouring nodes connected to find the optimal path. This consequently increases the communication burden, which might be against the main strength of the original randomised gossip, at each iteration. 

Based on these observations, this paper proposes a new variant of gossip process for distributed network-wide average computation. The algorithm developed is called sample greedy gossip (SGG) since it exploits stochastic sampling strategy. Unlike previous gossip algorithms, the proposed algorithm is a generalised version of randomised gossip and greedy gossip: it exhibits positive features of both randomised and greedy gossip. Different from randomised gossip, SGG takes the advantage of greedy gossip by applying the greedy node selection strategy to a randomly active node set. As the expected size of the active node set is smaller than the number of locally-connected neighbours, SGG can reduce the communication overload in the average sense. Therefore, the proposed SGG algorithm provides great flexibility and well balance between communication cost and convergence performance. Theoretical analysis reveals that the proposed SGG algorithm guarantees asymptotic convergence to the average state. The effect of the stochastic sampling strategy on the convergence rate of the proposed SGG algorithm is also theoretically analysed to demonstrate the tradeoff performance of SGG. To provide better insights into SGG, we derive the upper bound of its convergence rate, which relates the convergence performance of SGG to that of randomised gossip and greedy gossip. 

The rest of the paper is organised as follows. Sec. \ref{sec:2} introduces of the proposed SGG algorithm. Sec. \ref{sec:3} presents theoretical analysis of the SGG algorithm, followed by some numerical evaluations provided in Sec. \ref{sec:5}. Finally, some conclusions are offered. 


\section{Algorithm Development}
\label{sec:2}

This section will develop a new SGG algorithm, which is a generalised version of the gossip algorithm, by exploiting the benefits of both randomised gossip \cite{boyd2006randomized} and greedy gossip\cite{ustebay2010greedy}.  We consider a network of $N$ agents and model the network topology as an undirected graph ${\mathcal G} = \left({\mathcal V},{\mathcal E}\right)$, where ${\mathcal V}=\left\{1,2,\cdots,N\right\}$ represents the node set and ${\mathcal E} \subset  {\mathcal V} \times {\mathcal V}$ denotes the edge set. If agents $s$ and $t$ can directly communicate with each other, then $(s,t) \in \mathcal E$. We assume the network is strongly-connected\footnote{Note that this condition is required for guaranteeing the asymptotic convergence of gossip algorithms. However, this dose not mean that gossip algorithms are only applicable to strongly-connected networks.}, e.g., any two nodes can communicate with each other through a multi-hop path. Denote $a_s$ as the available information from the $s$th node and is initialised as $a_s(0)$. The objective of gossip algorithm is to enforce all local nodes to make an agreement on the initial average value $\sum\nolimits_{s = 1}^{N} {{a_s}\left( 0 \right)}$ using only local information available from the connected neighbours. 

At the $l$th round of the randomised gossip iterations, a node $s \in \mathcal V$ is randomly selected for communication. This can be accomplished by using an asynchronous time mode\footnote{Each node has a Poisson time clock with rate 1. At every gossip iteration, only one node is randomly picked up to communicate with one of its neighbours by clock ticking.} , as described in \cite{boyd2006randomized}. Then, node $s$ randomly selects a neighbour node $t \in {\cal N}_s$, where ${\cal N}_s$ represents the set of the nodes connected to the $s$th node (not including $s$), and performs information average as
\begin{equation}
a_s\left( l \right) = a_t\left( l \right) = \frac{a_s\left( l-1 \right) + a_t\left( l-1 \right)}{2}
\label{eq:3}
\end{equation}

For strongly-connected network, it was shown in \cite{boyd2006randomized} that the randomised gossip algorithm guarantees asymptotic convergence to the initial average value, that is $\mathop {\lim }\limits_{l \to \infty } {a_s\left( l \right)} = \frac{1}{N}\sum\nolimits_{s = 1}^{N} {{a_s}\left( 0 \right)}$. The main advantage of randomised gossip is that it has relatively low communication burden since each node only needs to communicate with one connected node during one iteration round. However, the randomised gossip algorithm might suffers from low convergence speed due to the randomised communication strategy. For this reason, a greedy gossip algorithm was proposed in \cite{ustebay2010greedy} to accelerate the convergence rate, i.e., reduce the number of required gossip iterations to reach an agreement among all nodes. The greedy gossip algorithm employs a deterministic procedure to select the communication nodes, i.e., active nodes. More specifically, at the $l$th round of the gossip iterations, the $s$th node selects a node $t^* \in {\cal N}_s$ that has the largest information discrepancy, that is
\begin{equation}
t^* =  \mathop {\max }\limits_{t \in {\mathcal N}_s} \left[a_{s}\left( l \right) - a_{t}\left( l \right)\right]^2
\label{eq:5}
\end{equation}

After finding $t^*$, the two agents perform information averaging using Eq. (\ref{eq:3}). Compared with the original randomised gossip process, the greedy gossip algorithm is proved to provide improved convergence speed. However, this strategy requires each node to communicate with all nodes connected to find the optimal communication path and thus might increase the communication overhead. 

Motivated by the complementary properties of these two different gossip algorithms, this paper develops a new variant of gossip algorithm, called SGG. The proposed gossip algorithm enjoys the advantages of both randomised gossip and greedy gossip: relatively faster convergence speed and lower communication burden. The fundamental idea of SGG is to apply the greedy node selection strategy to the set of randomly active nodes, ${\mathcal A}_s$, which is a subset of ${\mathcal N}_s$, i.e., ${\mathcal A}_s \subset {\mathcal N}_{s}$.  For notational convenience, we define ${\mathcal N}_s \buildrel \Delta \over = \left\{ Sn(1),Sn(2),\cdots,Sn\left(\left|{\mathcal N}_s\right|\right)  \right\}$ with $Sn(i) \in  \left\{1,2,\cdots,N \right\}$ for all $i \in \left\{ 1,2,\cdots,\left|{\mathcal N}_s\right| \right\}$, where $\left|{\mathcal N}_s\right|$ denotes the cardinality of set ${\mathcal N}_s$. The active node set ${\mathcal A}_s$ is generated by randomly selecting nodes from $\mathcal N_{s}$ for communication based on a stochastic uniform sampling procedure. More specifically, all nodes in ${\mathcal N}_s$ generate a sample, i.e., $q_i$ for all $i \in \left\{1, 2, \ldots, \left|{\mathcal N}_s\right|\right\}$, from uniform distribution ${\mathcal U}\left( {0,1} \right)$. Each node from $\mathcal N_{s}$ is then activated with probability $p \in [0,1]$: if $q_i \le p$, node $Sn(i)$ decides to actively communicate with node $s$. Once we have ${\mathcal A}_s$, the $s$th node selects a suboptimal node for information exchange as
\begin{equation}
t^* =  \mathop {\max }\limits_{t \in {\mathcal A}_s} \left[a_{s}\left( l \right) - a_{t}\left( l \right)\right]^2
\end{equation}
and leverages Eq. (\ref{eq:3}) for information update. If no node has been activated by the sampling strategy, i.e., ${\mathcal A}_s = \emptyset$, we perform randomised gossip for update. The pseudo code of the proposed SGG algorithm is summarised in Algorithm \ref{algo:1}.

\begin{remark}
It can be clearly noted that the proposed algorithm only utilises a random subset of ${\mathcal N}_s$ to perform greedy node selection. The communication burden, therefore, depends on $p$, which determines the size of ${\mathcal A}_s$. For example, if the node activation probability $p$ is chosen as $p=0.5$, the communication burden is half in average sense at every gossip iteration, compared with the greedy gossip. Therefore, the proposed SGG can be viewed as a generalised version of the gossip algorithm that can trade-off between randomised gossip and greedy gossip: the proposed SGG becomes greedy gossip with $p=1$ and reduces to randomised gossip with $p=0$.
\end{remark}

\begin{algorithm}
\caption{Sample greedy gossip}\label{algo:1}
\textbf{Input:} Initial local information $a_s(0)$, maximum iteration step $L$, node activation probability $p$\\
\textbf{Output:} Fused information $a_s(L)$
\begin{algorithmic}[1]
\State Randomly selects a node $s$ from the network
\For {$l = 1:L$}
	\State ${\mathcal A}_s = \emptyset $
	\Comment{Initialise the active node set $ {\mathcal A}_s$}
	\For {$i = 1: \left|{\mathcal N}_s\right|$}
		\State Node $Sn(i)$ generates a sample $q_i \sim {\mathcal U}\left( {0,1} \right)$
		\If {$q_i \le p$}
			\State Node $Sn(i)$ decides to actively communicate with node $s$
			\State ${\mathcal A}_s = {\mathcal A}_s \cup Sn(i)$
			\Comment{Update the active node set $ {\mathcal A}_s$}
                      \EndIf
	\EndFor
	\If {${\mathcal A}_s \ne \emptyset $}
           	\State $t^* =  \mathop {\max }\limits_{t\in {\mathcal A}_s} \left[a_{s}\left( l-1 \right) - a_{t}\left( l-1 \right)\right]^2$
		\Comment{Greedy node selection from ${\mathcal A}_s$}
           	\State $a_s\left( l \right) = a_{t^*}\left( l \right) = \frac{1}{2}\left[a_s\left( l -1\right) + a_{t^*}\left( l-1 \right)\right]$
	\Else
		\State Randomly selects a node $t$ from $ {\mathcal N}_s$
		\Comment{Randomised gossip}
		\State $a_s\left( l \right) = a_{t}\left( l \right) = \frac{1}{2}\left[a_s\left( l -1\right) + a_{t}\left( l-1 \right)\right]$
	\EndIf

\EndFor
\end{algorithmic}
\end{algorithm}

\section{Convergence Analysis}
\label{sec:3}

In this section, we provide detailed theoretical analysis of the proposed SGG algorithm, demonstrating its convergence performance and comparing its convergence rate with that of randomised gossip and greedy gossip. Note that the analysis carried out in this section is based on the asynchronous time mode. However, the extension to synchronous time model\footnote{Unlike asynchronous time model, all nodes in every gossip iteration communicate
simultaneously with one of their neighbour.} is straightforward and the qualitative/quantitative results are unaffected by the type of time model \cite{kempe2003gossip,karp2000randomized,boyd2006randomized}.

\subsection{Asymptotic Convergence Analysis}

This subsection analyses the asymptotic convergence of the proposed SGG algorithm. For notational convenience, we denote $a(l)=\left[a_1(l),a_2(l),\cdots,a_N(l)\right]^T$ and define $\overline{a}$ as a column vector with each element being $\frac{1}{N}\sum\nolimits_{l = 1}^N {{a_l}\left( 0 \right)} $. The main results are presented in the following theorem.

\begin{theorem}
The proposed SGG algorithm guarantees asymptotic convergence to the average state $\bar a$.
\end{theorem}

\begin{proof}
Assume that agents $s$ and the $t$ perform gossip at the $l$th iteration of SGG, then the recursive update of SGG can be obtained as
\begin{equation} \label{eq:ana1} 
a(l)=a(l-1)-\frac{1}{2} g(l)
\end{equation} 
where $g(l) \in {\mathbb R}^N$ is a column vector with its elements being
\begin{equation} \label{eq:ana2} 
g_{i}(l)=\left\{
\begin{array}{cc}{a_{s}(l-1)-a_{t}(l-1),} & {\text { for } i=s} \\ 
{-\left(a_{s}(l-1)-a_{t}(l-1)\right),} & {\text { for } i=t} \\ 
{0,} & {\text { otherwise }}
\end{array}\right.
\end{equation} 

Based on Eq. (\ref{eq:ana1}), the recursive update of the squared error is determined as
\begin{equation} \label{eq:ana3} 
\begin{split}
\|a(l)-\overline{a}\|^{2} &=\left\|a(l-1)-\frac{1}{2} g(l)-\overline{a}\right\|^{2} \\ 
&=\|a(l-1)-\overline{a}\|^{2}-\left[a(l-1)-\overline{a}\right]^T g(l)+\frac{1}{4}\|g(l)\|^{2} \\
&=\|a(l-1)-\overline{a}\|^{2}-\frac{1}{2}\left[a_{s}(l-1)-a_{t}(l-1)\right]^{2}
\end{split}
\end{equation} 

Note that both $s$ and $t$ are random in the proposed SGG. For this reason, we will examine the expected squared error, i.e., $\mathbb{E}\left[\|a(l)-\overline{a}\|^{2}\right]$, in the following analysis. Taking the expectation on $\left[a_{s}(l-1)-a_{t}(l-1)\right]^{2}$ gives
\begin{equation} \label{eq:ana4} 
\begin{split}
\mathbb{E}\left\{\left[a_{s}(l-1)-a_{t}(l-1)\right]^{2}\right\}=&\frac{1}{N}\sum\limits_{s = 1}^N {\sum\limits_{m = 1}^{\left| \mathcal{N}_{s} \right|} {{p^m}{{\left( {1 - p} \right)}^{\left| \mathcal{N}_{s} \right| - m}}\sum\limits_{{\mathcal N_{s,m}} \in \left\{\mathcal N_{s,m}\right\}} {\mathop {\max }\limits_{t \in {\mathcal N_{s,m}}} {{\left[ {{a_s}\left( {l - 1} \right) - {a_t}\left( {l - 1} \right)} \right]}^2}} } }\\
& +\frac{1}{N}\sum\limits_{s = 1}^N {{{\left( {1 - p} \right)}^{\left| \mathcal{N}_{s} \right|}}\frac{1}{{\left| \mathcal{N}_{s} \right|}}\sum\limits_{t \in \mathcal{N}_{s}} {{\left[ {{a_s}\left( {l - 1} \right) - {a_t}\left( {l - 1} \right)} \right]}^2} } 
\end{split}
\end{equation}
where ${\mathcal N}_{s,m}$ denotes a set of $m$ nodes, randomly drawn from $\mathcal N_{s}$ and $\left\{\mathcal N_{s,m}\right\}$ stands for the set that includes all possible $\mathcal N_{s,m}$.

Note that the first term on the right hand side of Eq. (\ref{eq:ana4}) refers to the case where $1 \leq m \leq \left| \mathcal{N}_{s} \right|$ nodes decide to communicate with the $s$th node, and the second term implies the case where no node has been activated by the sampling procedure. As stated in Algorithm \ref{algo:1}, if no node decides to communicate with node $s$ during the sampling phase, we perform the randomised gossip for update. From Eq. (\ref{eq:ana4}), it is clear that $\mathbb{E}\left\{\left[a_{s}(l-1)-a_{t}(l-1)\right]^{2}\right\}\ge0$, where the equality holds if and only if $a(l-1) = \bar a$. This means that, unless all nodes make agreement on the average state $\bar a$, the proposed SGG algorithm will make progress in expectation towards the average state $\bar a$. Moreover, by repeatedly applying recursion (\ref{eq:ana4}) and taking the expectation, we have
\begin{equation} \label{eq:ana4a} 
\begin{split}
\mathbb{E}\left[\|a(l)-\overline{a}\|^{2} \right]&=\mathbb{E}\left[\|a(l-1)-\overline{a}\|^{2} \right]-\frac{1}{2}\mathbb{E}\left\{\left[a_{s}(l-1)-a_{t}(l-1)\right]^{2}\right\} \\
& = \mathbb{E}\left[\|a(0)-\overline{a}\|^{2} \right] - \frac{1}{2}\sum\limits_{i = 1}^l {\mathbb{E}\left\{\left[a_{s}(i-1)-a_{t}(i-1)\right]^{2}\right\}} 
\end{split}
\end{equation} 

Since $\mathbb{E}\left[\|a(l)-\overline{a}\|^{2} \right] \ge 0$, we have
\begin{equation} \label{eq:ana4b} 
\mathbb{E}\left[\|a(0)-\overline{a}\|^{2} \right] \ge \frac{1}{2}\sum\limits_{i = 1}^l {\mathbb{E}\left\{\left[a_{s}(i-1)-a_{t}(i-1)\right]^{2}\right\}}
\end{equation} 
which implies that $\mathbb{E}\left\{\left[a_{s}(l-1)-a_{t}(l-1)\right]^{2}\right\} \to 0$ as $l \to \infty$. As the solution $a(l) = \bar a$ is the only stationary point of stochastic recursion (\ref{eq:ana3}), the proposed SGG algorithm guarantees asymptotic convergence to the average state $\bar a$.
\end{proof}

\subsection{Convergence Rate Analysis}

This subsection analyses convergence rate of the proposed SGG algorithm. We first investigates the effect of node activation probability $p$ on the convergence speed and then derive the convergence bound. 

\subsubsection{Effect of Node Activation Probability $p$}
The tradeoff performance of the proposed SGG algorithm highly depends on the node activation probability $p$. For this reason, we will theoretically analyse the effect of node activation probability $p$ and the result is presented in the following theorem.

\begin{theorem}
The convergence rate of the proposed SGG algorithm increases with the increase of node activation probability.
\end{theorem}

\begin{proof}
Define
\begin{equation} \label{eq:anap1} 
\begin{split}
f_s(p)=&\sum\limits_{m = 1}^{\left| \mathcal{N}_{s} \right|} {{p^m}{{\left( {1 - p} \right)}^{\left| \mathcal{N}_{s} \right| - m}}\sum\limits_{{\mathcal N_{s,m}} \in \left\{\mathcal N_{s,m}\right\}} {\mathop {\max }\limits_{t \in {\mathcal N_{s,m}}} {{\left[ {{a_s}\left( {l - 1} \right) - {a_t}\left( {l - 1} \right)} \right]}^2}} } \\
& +{{\left( {1 - p} \right)}^{\left| \mathcal{N}_{s} \right|}}\frac{1}{{\left| \mathcal{N}_{s} \right|}}\sum\limits_{t \in \mathcal{N}_{s}} {{\left[ {{a_s}\left( {l - 1} \right) - {a_t}\left( {l - 1} \right)} \right]}^2} 
\end{split}
\end{equation}
From Eqs. (\ref{eq:ana4a}) and (\ref{eq:ana4}), it is clear that the expected reduction of the square is given by
\begin{equation} \label{eq:analp1a}
\mathbb{E}\left[\|a(l-1)-\overline{a}\|^{2} \right]-\mathbb{E}\left[\|a(l)-\overline{a}\|^{2} \right] =  \frac{1}{2N}\sum\limits_{s = 1}^N {f_s(p)}
\end{equation} 
Taking the partial derivative of $f_s(p)$ with respect to $p$ gives
\begin{equation} \label{eq:anap3} 
\begin{split}
\frac{\partial f_s(p)}{\partial p}=&\sum\limits_{m = 1}^{\left| \mathcal{N}_{s} \right|} {{p^{m-1}{{\left( {1 - p} \right)}^{\left| \mathcal{N}_{s} \right| - m-1}} \left( m -\left| \mathcal{N}_{s} \right| p \right)}\sum\limits_{{\mathcal N_{s,m}} \in \left\{\mathcal N_{s,m}\right\}} {\mathop {\max }\limits_{t \in {\mathcal N_{s,m}}} {{\left[ {{a_s}\left( {l - 1} \right) - {a_t}\left( {l - 1} \right)} \right]}^2}} } \\
& -{{{\left( {1 - p} \right)}^{\left| \mathcal{N}_{s} \right| -1}}}\sum\limits_{t \in \mathcal{N}_{s}} {{\left[ {{a_s}\left( {l - 1} \right) - {a_t}\left( {l - 1} \right)} \right]}^2}  \\  
\end{split}
\end{equation}

For a specific $m$, the number of possible $\mathcal N_{s,m}$ is given by
\begin{equation} \label{eq:ana5} 
C_{\left| {{\mathcal N}_s} \right|}^m = \frac{{\left| {{\mathcal N}_s} \right|!}}{{m!\left( {\left| {{\mathcal N}_s} \right| - m} \right)!}}
\end{equation} 
Now let us define $m^* = \lfloor {\left| \mathcal{N}_{s} \right| p} \rfloor $. Then, it is trivial that
\begin{equation} \label{eq:anap4} 
\begin{split}
&( m -\left| \mathcal{N}_{s} \right| p )\sum\limits_{{\mathcal N_{s,m}}\in \left\{ \mathcal N_{s,m}\right\}} {\mathop {\max }\limits_{t \in {\mathcal N_{s,m}}} {{\left[ {{a_s}\left( {l - 1} \right) - {a_t}\left( {l - 1} \right)} \right]}^2}} \\
 &\quad \quad \quad \quad \ge  ( m -\left| \mathcal{N}_{s} \right| p )\frac{C_{\left| \mathcal{N}_{s} \right|}^{m}}{C_{\left| \mathcal{N}_{s} \right|}^{m^*}}  \sum\limits_{{\mathcal N_{s,m^*}}\in \left\{ \mathcal N_{s,m^*}\right\}} {\mathop {\max }\limits_{t \in {\mathcal N_{s,m^*}}} {{\left[ {{a_s}\left( {l - 1} \right) - {a_t}\left( {l - 1} \right)} \right]}^2}} 
 \end{split}
\end{equation}
\begin{equation} \label{eq:anap5} 
\begin{split}
&\sum\limits_{t \in \mathcal{N}_{s}} {{\left[ {{a_s}\left( {l - 1} \right) - {a_t}\left( {l - 1} \right)} \right]}^2}  \\
&\quad \quad \quad \quad \le \frac{\left| \mathcal{N}_{s} \right|}{C_{\left| \mathcal{N}_{s} \right|}^{m^*}}\sum\limits_{{\mathcal N_{s,m^*}}\in \left\{ \mathcal N_{s,m^*}\right\}} {\mathop {\max }\limits_{t \in {\mathcal N_{s,m^*}}} {{\left[ {{a_s}\left( {l - 1} \right) - {a_t}\left( {l - 1} \right)} \right]}^2}} 
\end{split}
\end{equation}
Substituting Eqs. (\ref{eq:anap4}) and (\ref{eq:anap5}) into Eq. (\ref{eq:anap3}) yields 
\begin{equation} \label{eq:anap6} 
\begin{split}
\frac{\partial f_s(p)}{\partial p}
\ge&\sum\limits_{m = 1}^{\left| \mathcal{N}_{s} \right|} {{p^{m-1}{{\left( {1 - p} \right)}^{\left| \mathcal{N}_{s} \right| - m-1}} \left( m -\left| \mathcal{N}_{s} \right| p \right)}\frac{C_{\left| \mathcal{N}_{s} \right|}^{m}}{C_{\left| \mathcal{N}_{s} \right|}^{m^*}} \sum\limits_{{\mathcal N_{s,m^*}} \in \left\{\mathcal N_{s,m^*}\right\}} {\mathop {\max }\limits_{t \in {\mathcal N_{s,m^*}}} {{\left[ {{a_s}\left( {l - 1} \right) - {a_t}\left( {l - 1} \right)} \right]}^2}} } \\
& -{{{\left( {1 - p} \right)}^{\left| \mathcal{N}_{s} \right| -1}}}\frac{\left| \mathcal{N}_{s} \right|}{C_{\left| \mathcal{N}_{s} \right|}^{m^*}}\sum\limits_{{\mathcal N_{s,m^*}}\in \left\{ \mathcal N_{s,m^*}\right\}} {\mathop {\max }\limits_{t \in {\mathcal N_{s,m^*}}} {{\left[ {{a_s}\left( {l - 1} \right) - {a_t}\left( {l - 1} \right)} \right]}^2}}  \\
=&\sum\limits_{m = 0}^{\left| \mathcal{N}_{s} \right|} {{mp^{m}{{\left( {1 - p} \right)}^{\left| \mathcal{N}_{s} \right| - m}} }\frac{C_{\left| \mathcal{N}_{s} \right|}^{m}}{p(1-p)C_{\left| \mathcal{N}_{s} \right|}^{m^*}} \sum\limits_{{\mathcal N_{s,m^*}} \in \left\{\mathcal N_{s,m^*}\right\}} {\mathop {\max }\limits_{t \in {\mathcal N_{s,m^*}}} {{\left[ {{a_s}\left( {l - 1} \right) - {a_t}\left( {l - 1} \right)} \right]}^2}} } \\
&-\sum\limits_{m = 0}^{\left| \mathcal{N}_{s} \right|} {{\left| \mathcal{N}_{s} \right|p^{m}{{\left( {1 - p} \right)}^{\left| \mathcal{N}_{s} \right| - m}} }\frac{C_{\left| \mathcal{N}_{s} \right|}^{m}}{(1-p)C_{\left| \mathcal{N}_{s} \right|}^{m^*}} \sum\limits_{{\mathcal N_{s,m^*}} \in \left\{\mathcal N_{s,m^*}\right\}} {\mathop {\max }\limits_{t \in {\mathcal N_{s,m^*}}} {{\left[ {{a_s}\left( {l - 1} \right) - {a_t}\left( {l - 1} \right)} \right]}^2}} } \\
\end{split}
\end{equation}
Since
\begin{equation} \label{eq:anap7} 
\begin{split}
&\sum\limits_{m = 0}^{\left| \mathcal{N}_{s} \right|} {mC_{\left| \mathcal{N}_{s} \right|}^{m}p^{m}{{\left( {1 - p} \right)}^{\left| \mathcal{N}_{s} \right| - m}} } =\left| \mathcal{N}_{s} \right|p\\
&\sum\limits_{m = 0}^{\left| \mathcal{N}_{s} \right|} {C_{\left| \mathcal{N}_{s} \right|}^{m}p^{m}{{\left( {1 - p} \right)}^{\left| \mathcal{N}_{s} \right| - m}} } =1
\end{split}
\end{equation}
evaluating Eq. (\ref{eq:anap6}) using Eq. (\ref{eq:anap7}) gives
\begin{equation} \label{eq:anap8} 
\frac{\partial f_s(p)}{\partial p}\ge0
\end{equation}
where the equality holds if and only if $a(l-1)=\bar a$. From Eqs. (\ref{eq:analp1a}) and (\ref{eq:anap8}), we have
\begin{equation} \label{eq:analp7a}
\frac{\partial \left( \mathbb{E}\left[\|a(l-1)-\overline{a}\|^{2} \right]-\mathbb{E}\left[\|a(l)-\overline{a}\|^{2} \right] \right)}{\partial p} =  \frac{1}{2N}\sum\limits_{s = 1}^N \frac{{\partial {f_s(p)}}}{\partial p} \ge 0
\end{equation} 
where the equality again holds if and only if $a(l-1)=\bar a$.
\end{proof}

\begin{remark}
Theorem 2 demonstrates that the proposed SGG algorithm provides faster convergence rate with higher node activation probability. This clearly reveals that SGG is a tradeoff between randomised gossip $(p=0)$ and greedy gossip $(p=1)$.  In viewing of this fact, it can be concluded that the proposed SGG algorithm provides great flexibility and well balance between communication cost and convergence performance. If enough resource is available for communication, then a higher node activation probability can be chosen to increase the convergence rate; otherwise, a relatively small value of $p$ is desirable.
\end{remark}

\subsubsection{Convergence Bound}
To analyse the bound on the convergence rate, we will first derive the lower and upper convergence bounds of each SGG iteration. The result is established in the following lemma.
\begin{lemma}
The expected reduction of the square error of each SGG iteration is lower and upper bounded by
\begin{equation} \label{eq:anal1a} 
\mathbb{E}\left[\|a(l-1)-\overline{a}\|^{2} \right]-\mathbb{E}\left[\|a(l)-\overline{a}\|^{2} \right] \ge  \frac{1}{2N}\sum\limits_{s = 1}^N {\frac{1}{{\left| {{{\mathcal N}_s}} \right|}}\sum\limits_{t \in {{\mathcal N}_s}} {{{\left[ {{a_s}\left( {l - 1} \right) - {a_t}\left( {l - 1} \right)} \right]}^2}} }  
\end{equation} 
\begin{equation} \label{eq:anal1b} 
\mathbb{E}\left[\|a(l-1)-\overline{a}\|^{2} \right]-\mathbb{E}\left[\|a(l)-\overline{a}\|^{2} \right]  \le \frac{1}{2N}\sum\limits_{s = 1}^N {\mathop {\max }\limits_{t \in {\mathcal N_{s}}} {{\left[ {{a_s}\left( {l - 1} \right) - {a_t}\left( {l - 1} \right)} \right]}^2}  }
\end{equation} 
where the equality holds if and only if $a(l-1)=\bar a$.
\end{lemma}

\begin{proof}
Note that the probability that $\mathcal N_{s,m}$ contains one specific node $t \in \mathcal N_{s}$ is
\begin{equation} \label{eq:ana6} 
\Pr  \left(t \in \mathcal N_{s}  \cup \mathcal N_{s,m} \right) = \frac{m}{\left| \mathcal{N}_{s} \right|}
\end{equation} 
Let $M_{s,m}^t$ be the number of occurrences of one specific node $t \in \mathcal N_{s}$ that appears in $\left\{\mathcal N_{s,m}\right\}$\footnote{ Here, we provide an example to explain the notation $M_{s,m}^t$. Consider ${\mathcal N}_s = \left\{ 1,2,3,4 \right\}$ and $m=2$. From the definition of ${\mathcal N}_{s,m}$, it can be concluded that ${\mathcal N}_{s,2}$ in this example has $C_4^2 = 6$ possibilities, i.e., $\left\{ 1,2\right\}$, $\left\{ 1,3\right\}$, $\left\{ 1,4\right\}$,  $\left\{ 2,3\right\}$, $\left\{ 2,4\right\}$, $\left\{ 3,4\right\}$. Therefore, we have $\left\{{\mathcal N}_{s,2}\right\} = \left\{ \left\{ 1,2\right\}, \left\{ 1,3\right\}, \left\{ 1,4\right\}, \left\{ 2,3\right\}, \left\{ 2,4\right\}, \left\{ 3,4\right\} \right\}$ and the number of occurrences of each node index in $\left\{{\mathcal N}_{s,2}\right\}$ is given by $M_{s,2}^1=M_{s,2}^2=M_{s,2}^3=M_{s,2}^4=\frac{2}{4}C_4^2=3$.},  which is given by
\begin{equation} \label{eq:ana7} 
M_{s,m}^t=\Pr\left(t \in \mathcal N_{s}  \cup \mathcal N_{s,m} \right)C_{\left| {{\mathcal N}_s} \right|}^m = \frac{m}{\left| \mathcal{N}_{s} \right|} C_{\left| {{\mathcal N}_s} \right|}^m
\end{equation} 
Using Eq. (\ref{eq:ana7}), we obtain
\begin{equation} \label{eq:ana8} 
\begin{split}
\sum\limits_{{\mathcal N_{s,m}}\in \left\{ \mathcal N_{s,m}\right\}} {\mathop {\max }\limits_{t \in {\mathcal N_{s,m}}} {{\left[ {{a_s}\left( {l - 1} \right) - {a_t}\left( {l - 1} \right)} \right]}^2}}
 & \ge \sum\limits_{{\mathcal N_{s,m}}\in \left\{ \mathcal N_{s,m}\right\}} {\frac{1}{m}\sum\limits_{t \in {\mathcal N_{s,m}}} {{{\left[ {{a_s}\left( {l - 1} \right) - {a_t}\left( {l - 1} \right)} \right]}^2}} }  \\ 
 & = \frac{1}{m}\sum\limits_{t \in {\mathcal N_{s}}} {{{M_{s,m}^t}{\left[ {{a_s}\left( {l - 1} \right) - {a_t}\left( {l - 1} \right)} \right]}^2}}  \\ 
  &= \frac{1}{{\left| {{{\mathcal N}_s}} \right|}}C_{\left| {{\mathcal N_{s}}} \right|}^m\sum\limits_{t \in {{\mathcal N}_s}} {{{\left[ {{a_s}\left( {l - 1} \right) - {a_t}\left( {l - 1} \right)} \right]}^2}}  
\end{split}
\end{equation}
Substituting Eq. (\ref{eq:ana8}) into Eq. (\ref{eq:ana4}) yields
\begin{equation} \label{eq:ana9} 
\begin{split}
\mathbb{E}\left\{\left[a_{s}(l-1)-a_{t}(l-1)\right]^{2}\right\}\ge&\frac{1}{N}\sum\limits_{s = 1}^N {\sum\limits_{m = 1}^{\left| {{\mathcal N_{s}}} \right|} {C_{\left| {{\mathcal N_{s}}} \right|}^m{p^m}{{\left( {1 - p} \right)}^{\left| {{\mathcal N_{s}}} \right| - m}}\frac{1}{{\left| {{\mathcal N_{s}}} \right|}}\sum\limits_{t \in {\mathcal N_{s}}} {{{\left[ {{a_s}\left( {l - 1} \right) - {a_t}\left( {l - 1} \right)} \right]}^2}} } } \\
& +\frac{1}{N}\sum\limits_{s = 1}^N {{{\left( {1 - p} \right)}^{\left| \mathcal{N}_{s} \right|}}\frac{1}{{\left| \mathcal{N}_{s} \right|}}\sum\limits_{t \in \mathcal{N}_{s}} {{\left[ {{a_s}\left( {l - 1} \right) - {a_t}\left( {l - 1} \right)} \right]}^2} } \\ 
  = &\frac{1}{N}\sum\limits_{s = 1}^N {\frac{1}{{\left| {{{\mathcal N}_s}} \right|}}\sum\limits_{t \in {{\mathcal N}_s}} {{{\left[ {{a_s}\left( {l - 1} \right) - {a_t}\left( {l - 1} \right)} \right]}^2}} \sum\limits_{m = 0}^{\left| {{{\mathcal N}_s}} \right|} {C_{\left| {{{\mathcal N}_s}} \right|}^m{p^m}{{\left( {1 - p} \right)}^{\left| {{{\mathcal N}_s}} \right| - m}}} }  \\ 
  = &\frac{1}{N}\sum\limits_{s = 1}^N {\frac{1}{{\left| {{{\mathcal N}_s}} \right|}}\sum\limits_{t \in {{\mathcal N}_s}} {{{\left[ {{a_s}\left( {l - 1} \right) - {a_t}\left( {l - 1} \right)} \right]}^2}} }  
\end{split}
\end{equation}
where the equality holds if and only if $a(l-1)=\bar a$.

Also, $\mathbb{E}\left\{\left[a_{s}(l-1)-a_{t}(l-1)\right]^{2}\right\}$ can be upper bounded by
\begin{equation} \label{eq:ana9a} 
\begin{split}
\mathbb{E}\left\{\left[a_{s}(l-1)-a_{t}(l-1)\right]^{2}\right\}\le&\frac{1}{N}\sum\limits_{s = 1}^N {\sum\limits_{m = 1}^{\left| \mathcal{N}_{s} \right|} {C_{\left| {{\mathcal N_{s}}} \right|}^m{p^m}{{\left( {1 - p} \right)}^{\left| \mathcal{N}_{s} \right| - m}} \mathop {\max }\limits_{t \in {\mathcal N_{s}}} {{\left[ {{a_s}\left( {l - 1} \right) - {a_t}\left( {l - 1} \right)} \right]}^2} } }\\
& +\frac{1}{N}\sum\limits_{s = 1}^N {{{\left( {1 - p} \right)}^{\left| \mathcal{N}_{s} \right|}}\mathop {\max }\limits_{t \in {\mathcal N_{s}}} {{\left[ {{a_s}\left( {l - 1} \right) - {a_t}\left( {l - 1} \right)} \right]}^2} } \\
=&\frac{1}{N}\sum\limits_{s = 1}^N {\mathop {\max }\limits_{t \in {\mathcal N_{s}}} {{\left[ {{a_s}\left( {l - 1} \right) - {a_t}\left( {l - 1} \right)} \right]}^2}  \sum\limits_{m = 0}^{\left| {{{\mathcal N}_s}} \right|} {C_{\left| {{{\mathcal N}_s}} \right|}^m{p^m}{{\left( {1 - p} \right)}^{\left| {{{\mathcal N}_s}} \right| - m}}}}\\
=&\frac{1}{N}\sum\limits_{s = 1}^N {\mathop {\max }\limits_{t \in {\mathcal N_{s}}} {{\left[ {{a_s}\left( {l - 1} \right) - {a_t}\left( {l - 1} \right)} \right]}^2}  }
\end{split}
\end{equation}
where the equality holds if and only if $a(l-1)=\bar a$. Combining the preceding two equations completes the proof.
\end{proof}

\begin{remark}
Note that the right hand side of inequality (\ref{eq:anal1a}) is the expected reduction of square error of randomised gossip and the right hand side of inequality (\ref{eq:anal1b}) is the expected reduction of square error of greedy gossip. Lemma 1 relates the performance of the proposed SGG to that of randomised gossip and greedy gossip. From Lemma 1, it is clear that the expected reduction of the square error of the proposed SGG algorithm is lower bounded by the randomised gossip and upper bounded by the greedy gossip. With this in mind, it can be concluded that the convergence rate of SGG is faster than the randomised gossip but slower than the greedy gossip in expectation.
\end{remark}

Now, let us directly examine the convergence speed of the proposed SGG algorithm in terms of the $\epsilon$-average time \cite{boyd2006randomized},
\begin{equation}
T_{ave}(\epsilon)=\sup _{a(0) \neq 0} \inf \left\{l : \operatorname{Pr}\left(\frac{\|a(l)-\overline{a}\|}{\|a(0)-\overline{a}\|} \geq \epsilon\right) \leq \epsilon\right\}
\end{equation}

For notational convenience, let $W^{SGG}\left(l\right)$ be the update transition matrix of implementing SGG at the $k$th iteration of SGG, e.g., $a\left(l\right)=W^{SGG}\left(l\right) a\left(l-1\right)$, and $W^{SGG}\left(1 : l\right)=\prod_{j=1}^{l} W^{SGG}\left(j\right)$ be the successive application of $l$ SGG updates. Likewise, we denote the one-step update matrices of randomised gossip as $W^{RG}\left(l\right)$ and define the application of $k$ successive randomised gossip updates  as $W^{RG}\left(1 : l\right)=\prod_{j=1}^{l} W^{RG}\left(j\right)$. We further denote $  \overline{W}=\mathbb{E}\left[W^{RG}\left(l\right)\right]  $ as the expected value of the one-step update matrix in the randomised gossip, and let $\lambda_2\left(\overline{W}\right)$ be  the second largest eigenvalue of $\overline{W}$. The following theorem establishes the theoretical upper bound of $\epsilon$-average time of the proposed SGG algorithm.

\begin{theorem}
The $\epsilon$-average time of the proposed SGG algorithm is upper bounded by
\begin{equation} \label{eq:anathsgg} 
T_{ave}^{SGG}(\epsilon) \leq \frac{3 \log \epsilon^{-1}}{\log{\left(\lambda_{2}(\overline{W})-\min\limits_{i \in [l]}\left\{\beta_{i}\right\}\right)^{-1}}}
\end{equation}
where 
\begin{equation} 
\beta_l = \frac{\eta_l}{\mathbb{E}\left[\left\|W^{SGG}(1 : l-1) a(0)-\overline{a}\right\|^{2}\right]} \ge 0
\end{equation}
with $\eta_l = 0$ if $a(l-1)=\bar a$, and otherwise,
\begin{equation} \label{eq:anathsgg1} 
\begin{split}
\eta_l = &\frac{1}{2N}\sum\limits_{s = 1}^N {\sum\limits_{m = 1}^{\left| \mathcal{N}_{s} \right|} {{p^m}{{\left( {1 - p} \right)}^{\left| \mathcal{N}_{s} \right| - m}}\sum\limits_{{\mathcal N_{s,m}} \in \left\{\mathcal N_{s,m}\right\}} {\mathop {\max }\limits_{t \in {\mathcal N_{s,m}}} {{\left[ {{a_s}\left( {l - 1} \right) - {a_t}\left( {l - 1} \right)} \right]}^2}} } }\\
& +\frac{1}{2N}\sum\limits_{s = 1}^N {{{\left( {1 - p} \right)}^{\left| \mathcal{N}_{s} \right|}}\frac{1}{{\left| \mathcal{N}_{s} \right|}}\sum\limits_{t \in \mathcal{N}_{s}} {{\left[ {{a_s}\left( {l - 1} \right) - {a_t}\left( {l - 1} \right)} \right]}^2} } \\
&-\frac{1}{2N} \sum_{s=1}^{n} \frac{1}{\left|\mathcal{N}_{s}\right|} \sum_{t \in \mathcal{N}_{s}}\left(a_{s}(l-1)-a_{t}(l-1)\right)^{2}
\end{split}
\end{equation}  
\end{theorem}

\begin{proof}
The recursive reduction of the expected square error of randomised gossip satisfies \cite{boyd2006randomized}
\begin{equation} \label{eq:ana11} 
\begin{split}
&\mathbb{E}\left[\left\|W^{RG}(1 : l) a(0)-\overline{a}\right\|^{2}\right] =\mathbb{E}\left[\left\|W^{R G}(1 : l-1) a(0)-\overline{a}\right\|^{2}\right]\\ 
&\quad \quad \quad -\frac{1}{2N} \sum_{s=1}^{n} \frac{1}{\left|\mathcal{N}_{s}\right|} \sum_{t \in \mathcal{N}_{s}}\left(a_{s}(l-1)-a_{t}(l-1)\right)^{2}\\ 
 &\quad \quad \leq \lambda_{2}(\overline{W}) \mathbb{E}\left[\left\|W^{R G}(1 : l-1) a(0)-\overline{a}\right\|^{2}\right]
\end{split}
\end{equation} 
By implementing one step randomised gossip after applying $l-1$ steps of SGG updates and using the relationship of Eq. (\ref{eq:ana11}), we have 
\begin{equation} \label{eq:ana12} 
\begin{split}
&\mathbb{E}\left[\left\|W^{RG}(k) W^{SGG}(1 : l-1) a(0)-\overline{a}\right\|^{2}\right]=\mathbb{E}\left[\left\|W^{SGG}(1 : l-1) a(0)-\overline{a}\right\|^{2}\right]  \\ 
&\quad \quad \quad -\frac{1}{2N} \sum_{s=1}^{n} \frac{1}{\left|\mathcal{N}_{s}\right|} \sum_{t \in \mathcal{N}_{s}}\left(a_{s}(l-1)-a_{t}(l-1)\right)^{2}\\ 
& \quad \quad \leq \lambda_{2}(\overline{W}) \mathbb{E}\left[\left\|W^{SGG}(1 : l-1) a(0)-\overline{a}\right\|^{2}\right]
\end{split}
\end{equation} 
Using Eq. (\ref{eq:ana12}), the convergence of the proposed SGG can be obtained as
\begin{equation} \label{eq:ana13} 
\begin{split}
&\mathbb{E}\left[\left\|W^{SGG}(1 : l) a(0)-\overline{a}\right\|^{2}\right] \\ 
&=\mathbb{E}\left[\left\|W^{SGG}(1 : l-1) a(0)-\overline{a}\right\|^{2}\right] -\frac{1}{2N} \sum_{s=1}^{n} \frac{1}{\left|\mathcal{N}_{s}\right|} \sum_{t \in \mathcal{N}_{s}}\left(a_{s}(l-1)-a_{t}(l-1)\right)^{2}\\ 
&\quad \quad -\frac{1}{2N}\sum\limits_{s = 1}^N {\sum\limits_{m = 1}^{\left| \mathcal{N}_{s} \right|} {{p^m}{{\left( {1 - p} \right)}^{\left| \mathcal{N}_{s} \right| - m}}\sum\limits_{{\mathcal N_{s,m}} \in \left\{\mathcal N_{s,m}\right\}} {\mathop {\max }\limits_{t \in {\mathcal N_{s,m}}} {{\left[ {{a_s}\left( {l - 1} \right) - {a_t}\left( {l - 1} \right)} \right]}^2}} } }\\
& \quad \quad -\frac{1}{2N}\sum\limits_{s = 1}^N {{{\left( {1 - p} \right)}^{\left| \mathcal{N}_{s} \right|}}\frac{1}{{\left| \mathcal{N}_{s} \right|}}\sum\limits_{t \in \mathcal{N}_{s}} {{\left[ {{a_s}\left( {l - 1} \right) - {a_t}\left( {l - 1} \right)} \right]}^2} } \\
&\quad \quad +\frac{1}{2N} \sum_{s=1}^{n} \frac{1}{\left|\mathcal{N}_{s}\right|} \sum_{t \in \mathcal{N}_{s}}\left(a_{s}(l-1)-a_{t}(l-1)\right)^{2} \\
& = \mathbb{E}\left[\left\|W^{SGG}(1 : l-1) a(0)-\overline{a}\right\|^{2}\right] -\frac{1}{2N} \sum_{s=1}^{n} \frac{1}{\left|\mathcal{N}_{s}\right|} \sum_{t \in \mathcal{N}_{s}}\left(a_{s}(l-1)-a_{t}(l-1)\right)^{2} - \eta_l 
\end{split}
\end{equation} 
According to Lemma 1, it is clear that $\eta_l \ge 0$ and the equality holds if and only if $a(l-1)=\bar a$. With this in mind, substituting Eq. (\ref{eq:ana12}) into Eq. (\ref{eq:ana13}) gives the convergence bound of the proposed SGG as
\begin{equation} \label{eq:ana15} 
\begin{split}
&\mathbb{E}\left[\left\|W^{SGG}(1 : l) a(0)-\overline{a}\right\|^{2}\right] \\ 
&\quad \quad \quad \quad  \leq \left[\lambda_{2}(\overline{W})-\beta_{l}\right] \mathbb{E}\left[\left\|W^{SGG}(1 : l-1) a(0)-\overline{a}\right\|^{2}\right]
\end{split}
\end{equation} 
Repeatedly using Eq. (\ref{eq:ana15}) yields
\begin{equation} \label{eq:ana16} 
\begin{split}
\mathbb{E}\left[\left\|W^{SGG}(1 : l) a(0)-\overline{a}\right\|^{2}\right] &\leq\|a(0)-\overline{a}\|^{2} \prod_{i=1}^{l}\left(\lambda_{2}(\overline{W})-\beta_{i}\right)\\
 &\leq \|a(0)-\overline{a}\|^{2}\left(\lambda_{2}(\overline{W})-\min\limits_{i \in [l]}\left\{\beta_{i}\right\}\right)^l
\end{split}
\end{equation} 
where the equality holds if and only if $a(l-1)=\bar a$.

Using Eq. (\ref{eq:ana16}) with the help of Markov's inequality, we have
\begin{equation} \label{eq:ana16a} 
\begin{split}
\Pr\left(\frac{\|W^{SGG}(1 : l) a(0)-\overline{a}\|}{\|a(0)-\overline{a}\|} \geq \epsilon\right) &= \Pr \left(\frac{\|W^{SGG}(1 : l) a(0)-\overline{a}\|^2}{\|a(0)-\overline{a}\|^2} \geq \epsilon^2\right)\\
&\leq \frac{\mathbb{E}\left[\|W^{SGG}(1 : l) a(0)-\overline{a}\|^{2}\right]}{\epsilon^{2}\|a(0)-\overline{a}\|^{2}} \\
 &\le \epsilon^{-2} \left(\lambda_{2}(\overline{W})-\min\limits_{i \in [l]}\left\{\beta_{i}\right\}\right)^l
\end{split}
\end{equation} 
which means that the $\epsilon$-averaging time for SGG is upper bounded by
\begin{equation} \label{eq:ana16b} 
T_{ave}^{SGG}(\epsilon) \leq \frac{3 \log \epsilon^{-1}}{\log{\left(\lambda_{2}(\overline{W})-\min\limits_{i \in [l]}\left\{\beta_{i}\right\}\right)^{-1}}}
\end{equation}
\end{proof}

\begin{remark}
Theorem 3 provides a theoretical upper bound for the $\epsilon$-averaging time of the proposed SGG algorithm. Note from Eqs. (\ref{eq:anap1}) and (\ref{eq:anathsgg1})  that $\eta_l$ and $f_s(p)$ have the same monotonicity with respect to $p$. This verifies that the upper bound of the $\epsilon$-averaging time becomes smaller with higher node activation probability.
\end{remark}

\subsection{Convergence Rate Comparison}
This subsection will characterise the upper bound of $\epsilon$-averaging time of SGG algorithm with respect to other gossip algorithms to provide better understandings and insights of the proposed algorithm. Similar to previous subsection, we denote the one-step update matrices of greedy gossip as $W^{GG}\left(l\right)$; and define the application of $k$ successive greedy gossip updates  as $W^{GG}\left(1 : l\right)=\prod_{j=1}^{l} W^{GG}\left(j\right)$. The main results of this subsection are provided in the following theorem.

\begin{theorem}
The upper bound of $\epsilon$-average time, denoted as ${\mathcal U}(\epsilon)$, of the proposed SGG algorithm is smaller than that of the randomised gossip, but larger than that of the greedy gossip, i.e.,
\begin{equation}\label{eq:anath1}
{\mathcal U}_{GG}(\epsilon) \le {\mathcal U}_{SGG}(\epsilon) \le {\mathcal U}_{RG}(\epsilon) 
\end{equation} 
\end{theorem}

\begin{proof}
Implementing one step SGG after applying $l-1$ steps of greedy gossip updates and using Eq. (\ref{eq:ana15}), we have 
\begin{equation} \label{eq:ana17} 
\begin{split}
&\mathbb{E}\left[\left\|W^{SGG}(l) W^{GG}(1 : l-1) a(0)-\overline{a}\right\|^{2}\right] =\mathbb{E}\left[\left\|W^{GG}(1 : l-1) a(0)-\overline{a}\right\|^{2}\right]  \\
&\quad \quad \quad-\frac{1}{2N}\sum\limits_{s = 1}^N {\sum\limits_{m = 1}^{\left| \mathcal{N}_{s} \right|} {{p^m}{{\left( {1 - p} \right)}^{\left| \mathcal{N}_{s} \right| - m}}\sum\limits_{{\mathcal N_{s,m}} \in \left\{\mathcal N_{s,m}\right\}} {\mathop {\max }\limits_{t \in {\mathcal N_{s,m}}} {{\left[ {{a_s}\left( {l - 1} \right) - {a_t}\left( {l - 1} \right)} \right]}^2}} } }\\
&\quad \quad \quad - \frac{1}{2N}\sum\limits_{s = 1}^N {{{\left( {1 - p} \right)}^{\left| \mathcal{N}_{s} \right|}}\frac{1}{{\left| \mathcal{N}_{s} \right|}}\sum\limits_{t \in \mathcal{N}_{s}} {{\left[ {{a_s}\left( {l - 1} \right) - {a_t}\left( {l - 1} \right)} \right]}^2} } \\
&\quad \quad \leq \left[\lambda_{2}(\overline{W})-\beta_{l}\right] \mathbb{E}\left[\left\|W^{GG}(1 : l-1) a(0)-\overline{a}\right\|^{2}\right]
\end{split}
\end{equation} 
Note that the recursive reduction of the expected square error of greedy gossip satisfies \cite{ustebay2010greedy}
\begin{equation} \label{eq:ana17a} 
\begin{split}
&\mathbb{E}\left[\left\|W^{GG}(1 : l) a(0)-\overline{a}\right\|^{2}\right] =\mathbb{E}\left[\left\|W^{GG}(1 : l-1) a(0)-\overline{a}\right\|^{2}\right]\\ 
&\quad \quad \quad -\frac{1}{2N}\sum\limits_{s = 1}^N {\mathop {\max }\limits_{t \in {\mathcal N_{s}}} {{\left[ {{a_s}\left( {l - 1} \right) - {a_t}\left( {l - 1} \right)} \right]}^2}  }
\end{split}
\end{equation} 
From Eqs. (\ref{eq:ana17}) and (\ref{eq:ana17a}), the convergence rate of the greedy gossip can be represented as
\begin{equation} \label{eq:ana18} 
\begin{split} 
&\mathbb{E}\left[\left\|W^{GG}(1 : l) a(0)-\overline{a}\right\|^{2}\right] =\mathbb{E}\left[\left\|W^{GG}(1 : l-1) a(0)-\overline{a}\right\|^{2}\right] \\
&\quad \quad \quad-\frac{1}{2N}\sum\limits_{s = 1}^N {\sum\limits_{m = 1}^{\left| \mathcal{N}_{s} \right|} {{p^m}{{\left( {1 - p} \right)}^{\left| \mathcal{N}_{s} \right| - m}}\sum\limits_{{\mathcal N_{s,m}} \in \left\{\mathcal N_{s,m}\right\}} {\mathop {\max }\limits_{t \in {\mathcal N_{s,m}}} {{\left[ {{a_s}\left( {l - 1} \right) - {a_t}\left( {l - 1} \right)} \right]}^2}} } }\\
&\quad \quad \quad-\frac{1}{2N}\sum\limits_{s = 1}^N {{{\left( {1 - p} \right)}^{\left| \mathcal{N}_{s} \right|}}\frac{1}{{\left| \mathcal{N}_{s} \right|}}\sum\limits_{t \in \mathcal{N}_{s}} {{\left[ {{a_s}\left( {l - 1} \right) - {a_t}\left( {l - 1} \right)} \right]}^2} } \\
&\quad \quad \quad+\frac{1}{2N}\sum\limits_{s = 1}^N {\sum\limits_{m = 1}^{\left| \mathcal{N}_{s} \right|} {{p^m}{{\left( {1 - p} \right)}^{\left| \mathcal{N}_{s} \right| - m}}\sum\limits_{{\mathcal N_{s,m}} \in \left\{\mathcal N_{s,m}\right\}} {\mathop {\max }\limits_{t \in {\mathcal N_{s,m}}} {{\left[ {{a_s}\left( {l - 1} \right) - {a_t}\left( {l - 1} \right)} \right]}^2}} } }\\
&\quad \quad \quad+\frac{1}{2N}\sum\limits_{s = 1}^N {{{\left( {1 - p} \right)}^{\left| \mathcal{N}_{s} \right|}}\frac{1}{{\left| \mathcal{N}_{s} \right|}}\sum\limits_{t \in \mathcal{N}_{s}} {{\left[ {{a_s}\left( {l - 1} \right) - {a_t}\left( {l - 1} \right)} \right]}^2} } \\
&\quad \quad \quad-\frac{1}{2N}\sum\limits_{s = 1}^N {\mathop {\max }\limits_{t \in {\mathcal N_{s}}} {{\left[ {{a_s}\left( {l - 1} \right) - {a_t}\left( {l - 1} \right)} \right]}^2}  } \\
&\quad \quad \leq \left[\lambda_{2}(\overline{W})-\beta_{l}-\xi_l\right] \mathbb{E}\left[\left\|W^{GG}(1 : l-1) a(0)-\overline{a}\right\|^{2}\right]
\end{split}
\end{equation} 
where
\begin{equation}
\begin{split}
\gamma_l=&\frac{1}{2N}\sum\limits_{s = 1}^N {\mathop {\max }\limits_{t \in {\mathcal N_{s}}} {{\left[ {{a_s}\left( {l - 1} \right) - {a_t}\left( {l - 1} \right)} \right]}^2}  }\\
&-\frac{1}{2N}\sum\limits_{s = 1}^N {\sum\limits_{m = 1}^{\left| \mathcal{N}_{s} \right|} {{p^m}{{\left( {1 - p} \right)}^{\left| \mathcal{N}_{s} \right| - m}}\sum\limits_{{\mathcal N_{s,m}} \in \left\{\mathcal N_{s,m}\right\}} {\mathop {\max }\limits_{t \in {\mathcal N_{s,m}}} {{\left[ {{a_s}\left( {l - 1} \right) - {a_t}\left( {l - 1} \right)} \right]}^2}} } }\\
&-\frac{1}{2N}\sum\limits_{s = 1}^N {{{\left( {1 - p} \right)}^{\left| \mathcal{N}_{s} \right|}}\frac{1}{{\left| \mathcal{N}_{s} \right|}}\sum\limits_{t \in \mathcal{N}_{s}} {{\left[ {{a_s}\left( {l - 1} \right) - {a_t}\left( {l - 1} \right)} \right]}^2} } 
\end{split}
\end{equation} 
\begin{equation} 
\xi_l = \frac{\gamma_l}{\mathbb{E}\left[\left\|W^{GG}(1 : l-1) a(0)-\overline{a}\right\|^{2}\right]}
\end{equation} 
It follows from Lemma 1 that $\xi_l \ge 0$ and the equality holds if and only if $a(l-1)=\bar a$. Repeatedly applying Eq. (\ref{eq:ana18}) gives
\begin{equation} \label{eq:ana19} 
\begin{split}
\mathbb{E}\left[\left\|W^{GG}(1 : l) a(0)-\overline{a}\right\|^{2}\right] &\leq\|a(0)-\overline{a}\|^{2} \prod_{i=1}^{l}\left(\lambda_{2}(\overline{W})-\beta_{i}-\xi_i\right)\\
 &\leq \|a(0)-\overline{a}\|^{2}\left(\lambda_{2}(\overline{W})-\min\limits_{i \in [l]}\left\{\beta_{i}\right\}-\min\limits_{i \in [l]}\left\{\xi_{i}\right\}\right)^l
\end{split}
\end{equation} 
Using Eq. (\ref{eq:ana19}) with the help of Markov's inequality, we have
\begin{equation} \label{eq:ana19a} 
\begin{split}
\Pr\left(\frac{\|W^{GG}(1 : l) a(0)-\overline{a}\|}{\|a(0)-\overline{a}\|} \geq \epsilon\right) &= \Pr \left(\frac{\|W^{GG}(1 : l) a(0)-\overline{a}\|^2}{\|a(0)-\overline{a}\|^2} \geq \epsilon^2\right)\\
&\leq \frac{\mathbb{E}\left[\|W^{GG}(1 : l) a(0)-\overline{a}\|^{2}\right]}{\epsilon^{2}\|a(0)-\overline{a}\|^{2}} \\
 &\le \epsilon^{-2} \left(\lambda_{2}(\overline{W})-\min\limits_{i \in [l]}\left\{\beta_{i}\right\}-\min\limits_{i \in [l]}\left\{\xi_{i}\right\}\right)^l
\end{split}
\end{equation} 
which means that the $\epsilon$-averaging time for greedy gossip is upper bounded by
\begin{equation} \label{eq:ana19b} 
T_{ave}^{GG}(\epsilon) \leq \frac{3 \log \epsilon^{-1}}{\log{\left(\lambda_{2}(\overline{W})-\min\limits_{i \in [l]}\left\{\beta_{i}\right\}-\min\limits_{i \in [l]}\left\{\xi_{i}\right\}\right)^{-1}}}
\end{equation}
Recall that the upper bound for the $\epsilon$-averaging time of randomised gossip is given by \cite{boyd2006randomized}
\begin{equation} \label{eq:ana20} 
T_{ave}^{RG}(\epsilon) \leq \frac{3 \log \epsilon^{-1}}{\log{\lambda_{2}(\overline{W})^{-1}}}
\end{equation}
Combining Eqs. (\ref{eq:ana16b}), (\ref{eq:ana19b}) and (\ref{eq:ana20}), it is easy to verify Eq. (\ref{eq:anath1}), which completes the proof.
\end{proof}

\begin{remark}
It follows from Eq. (\ref{eq:anathsgg1}) that $\eta_l=0$ if $p=0$, which implies that ${\mathcal U}_{SGG}(\epsilon) = {\mathcal U}_{RG}(\epsilon)$ with zero node activation probability. Additionally, we can easily verify that $\gamma_l=0$ with unity node activation probability, i.e., $p=1$. This means that the convergence bound of the proposed SGG algorithm becomes the same as that of the greedy gossip, i.e., ${\mathcal U}_{SGG}(\epsilon) = {\mathcal U}_{GG}(\epsilon)$. These results also clearly demonstrate the tradeoff performance of SGG in terms of convergence bound and coincide with previous analysis.
\end{remark}

\begin{remark}
Note that Theorem 4 only provides theoretical comparison of the upper bound of $\epsilon$-averaging time for different gossip algorithms. However, extensive numerical analysis shows that the proposed algorithm achieves very comparable performance to that of the greedy gossip with significantly reduced communication burden.
\end{remark}

\begin{remark}
It is straightforward to verify that both $\beta_i$ and $\xi_i$ are trivial and can be ignored for sparse networks, compared to $\lambda_{2}\left(\overline{W}\right)$. To see this, let us consider a special network: each local node can only communicate with one local node. Under this condition, one can imply that $\beta_l=0$ and $\xi_l=0$, which indicates that ${\mathcal U}_{GG}(\epsilon) = {\mathcal U}_{SGG}(\epsilon) = {\mathcal U}_{RG}(\epsilon)$. This also means that the performance of both greedy gossip and the proposed SGG is very close to that of randomised gossip for sparse network and only provide significant convergence speed improvement for dense network. This fact will be illustrated via numerical investigations in the next section.
\end{remark}

\section{Numerical Simulations}
\label{sec:5}

This section evaluates performance the proposed SGG algorithm with comparison to randomised gossip and greedy gossip using Monte-Carlo simulations.  

As stated in \cite{gupta2000capacity,ustebay2010greedy}, the random geometric topology is widely-accepted as a general model to represent the connectivity of wireless networks for the purpose of analysing the characteristics of network-wide computation algorithms. For this reason, we choose the random geometric graph with 200 nodes as the network connection topology in performance evaluation. For the random geometric network, each agent is randomly placed inside a unit square region and two nodes are connected if their relative distance is less than $r\left(N\right)=\sqrt{\frac{d\log N}{N}}$ with $d$ being a scaling factor. It is clear that the random geometric network with higher $d$ provides dense connection while with lower $d$ represents sparse network topology. An example of the random geometric network topology with $d=2$ is presented in Fig. \ref{fig:topology}. Similar to \cite{ustebay2010greedy}, we utilise four different initial conditions to examine the performance of all algorithms: (1) Gaussian bumps; (2) linear-varying field; (3) spike field; and (4) uniform random field. The distributions of Gaussian bumps and linear-varying field are shown in Fig. \ref{fig:field}. The spike field  initialises all local agents by setting the value of one node as 1 and all others as 0. For the uniform random field, all nodes are initialised with a random value that is drawn from a normal distribution ${\mathcal N} \left(0,1\right)$. For performance evaluation of all tested gossip algorithms, we check the expected reduction of the relative error ${\left\| {a\left( L \right) - {\bar a}} \right\|}/{\left\| {a\left( 0 \right) - {\bar a}} \right\|}$ over 1000 Monte-Carlo runs.

\begin{figure*}[h!]
\centering
\includegraphics[width=.5\linewidth]{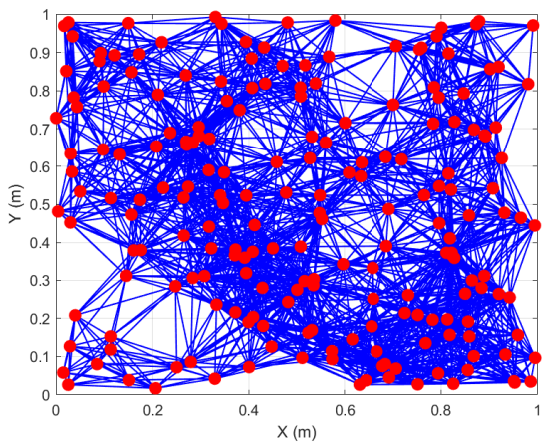}
\caption{An example of the random geometric network topology. The red circles denote the agent locations and the blue lines refer to the connections between local nodes. Each node is randomly placed inside the surveillance region and two nodes are connected if their relative distance is less than $\sqrt{\frac{2\log N}{N}}$.}
\label{fig:topology}
\end{figure*}

\begin{figure*}[h!]
\centering
\begin{tabular}{cc}
	\subfloat[Gaussian bumps]{\includegraphics[width=.5\linewidth]{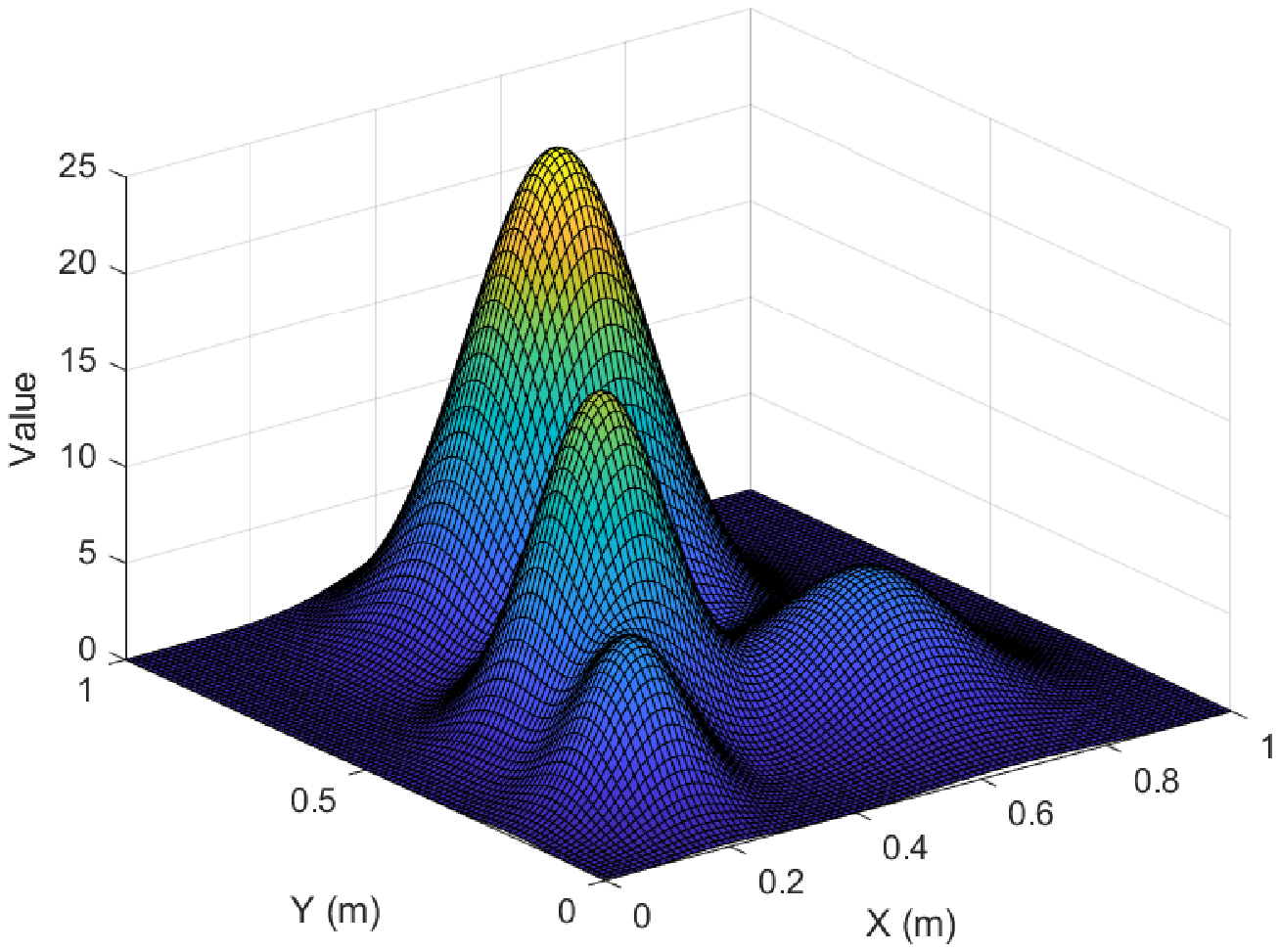}}
	&\subfloat[Linear-varying field]{\includegraphics[width=.5\linewidth]{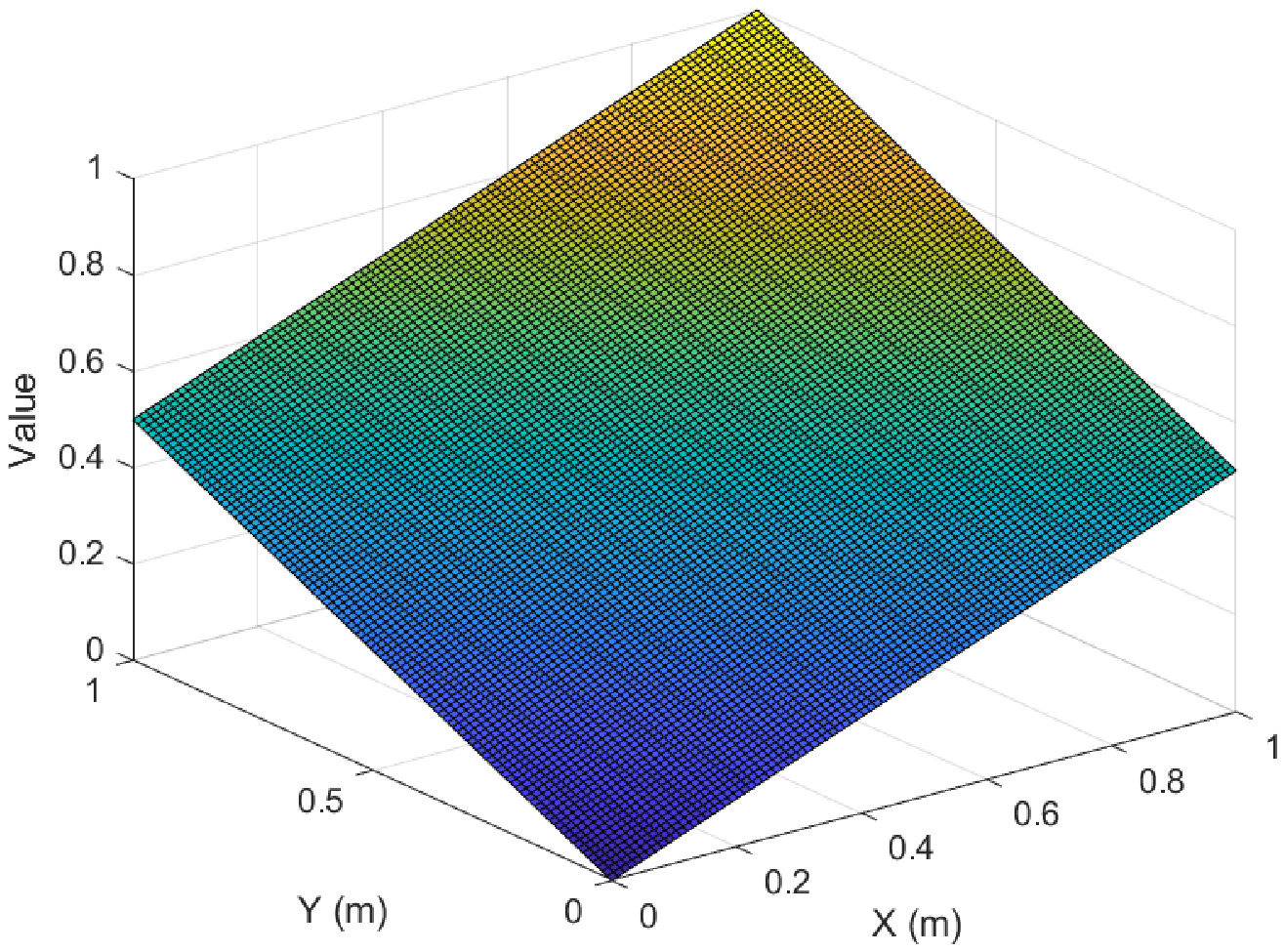}}
\end{tabular}
\caption{Distributions of Gaussian bumps and linear-varying field in the considered scenario.}
\label{fig:field}
\end{figure*}

\subsection{Effect of Communication Iterations}

In gossip-based distributed computation, information transmission via multiple rounds of communication among locally-connected nodes are required and the performance highly depends on the number of iterations, $L$. In order to investigate the effect of the parameter $L$ on the convergence speed, Monte-Carlo comparisons of different gossip algorithms are carried out with respect to different number of iterations. The main objective of the performance comparison in this subsection is to validate the analysis results of the gossip algorithm, which are presented in Sec. \ref{sec:3}. The simulation results of average convergence error obtained from Monte-Carlo simulations are depicted in Fig. \ref{fig:iteration}. In the simulations, the node activation probability for implementing SGG and the scaling factor $d$ that models the network connectivity are set as $p=0.5$ and $d=2$, respectively. 

From Fig. \ref{fig:iteration}, it can be noted that randomised gossip has the lowest convergence speed among these three different gossip algorithms. As greedy gossip process picks up the optimal communication path for every local node at each gossip iteration, it exhibits the fastest convergence rate at the expense of high communication burden. The proposed SGG only leverages a suboptimal communication path, e.g., performing greedy node selection within a  set of randomly-chosen active nodes. Therefore, the SGG provides tradeoff convergence performance between randomised gossip and greedy gossip. As the probability threshold in selecting the active node is $p=0.5$, the proposed algorithm only requires half communication burden in the average sense at each iteration, compared to the greedy gossip. These results conform with the analytic findings presented in Sec. \ref{sec:3}. Interestingly, the performance of SGG is very comparable to that of the greedy gossip and its convergence rate is much faster than that of the randomised gossip even with $p=0.5$. 

\begin{figure}[h!]
\centering
\begin{tabular}{cc}
	\subfloat[Gaussian bumps field]{\includegraphics[width=.5\linewidth]{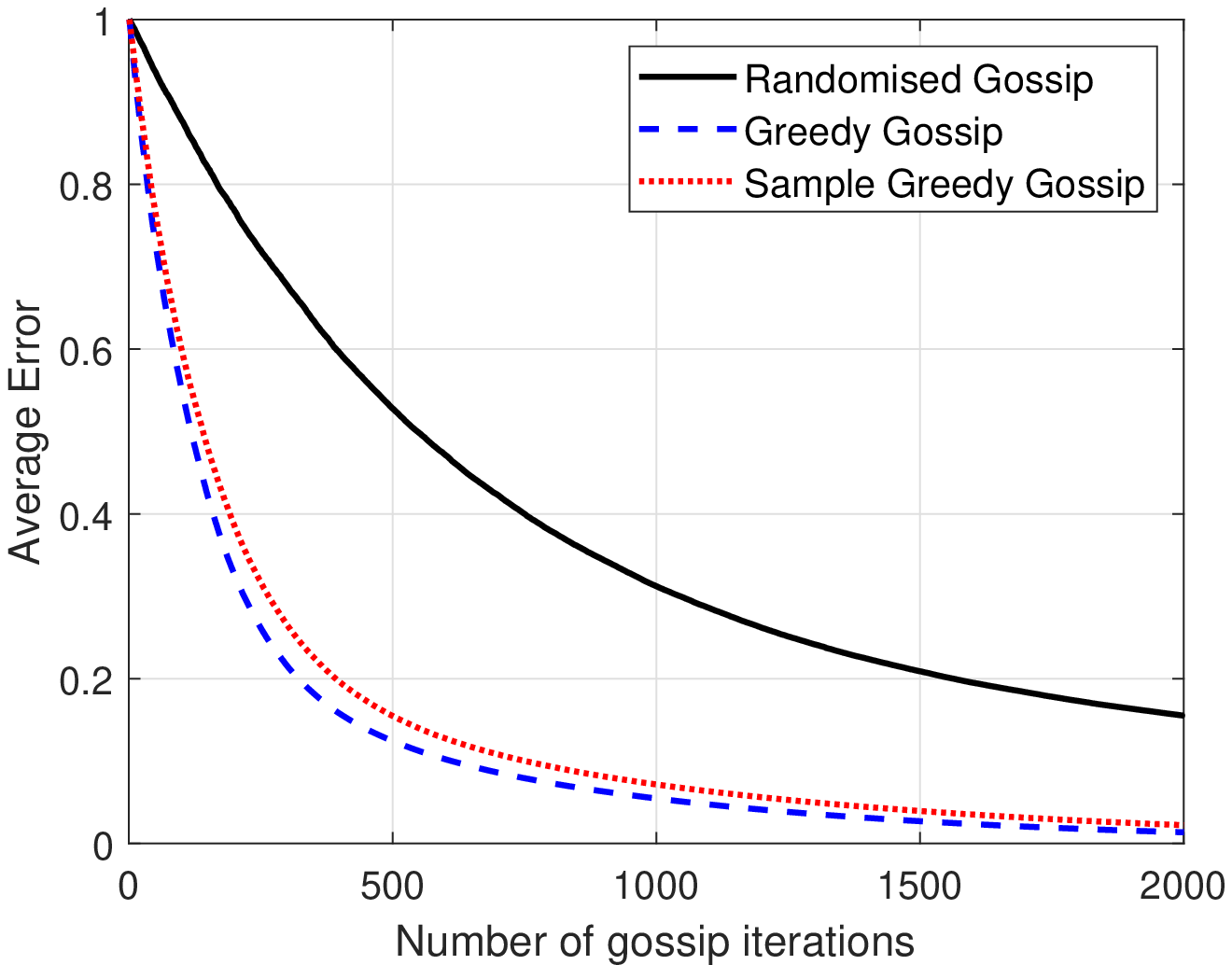}}
	&\subfloat[Spike field]{\includegraphics[width=.5\linewidth]{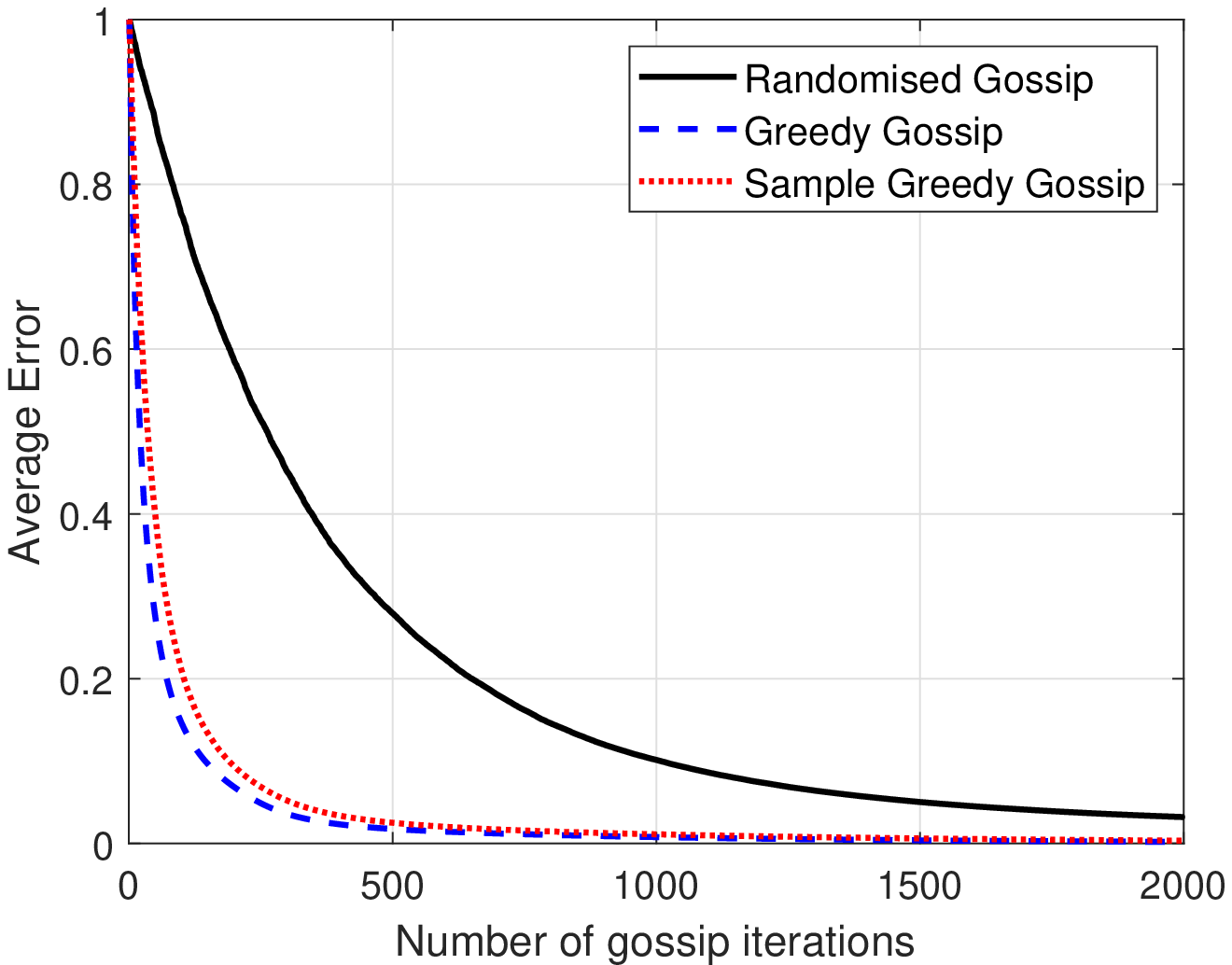}}
\end{tabular}
\begin{tabular}{cc}
	\subfloat[Uniform random field]{\includegraphics[width=.5\linewidth]{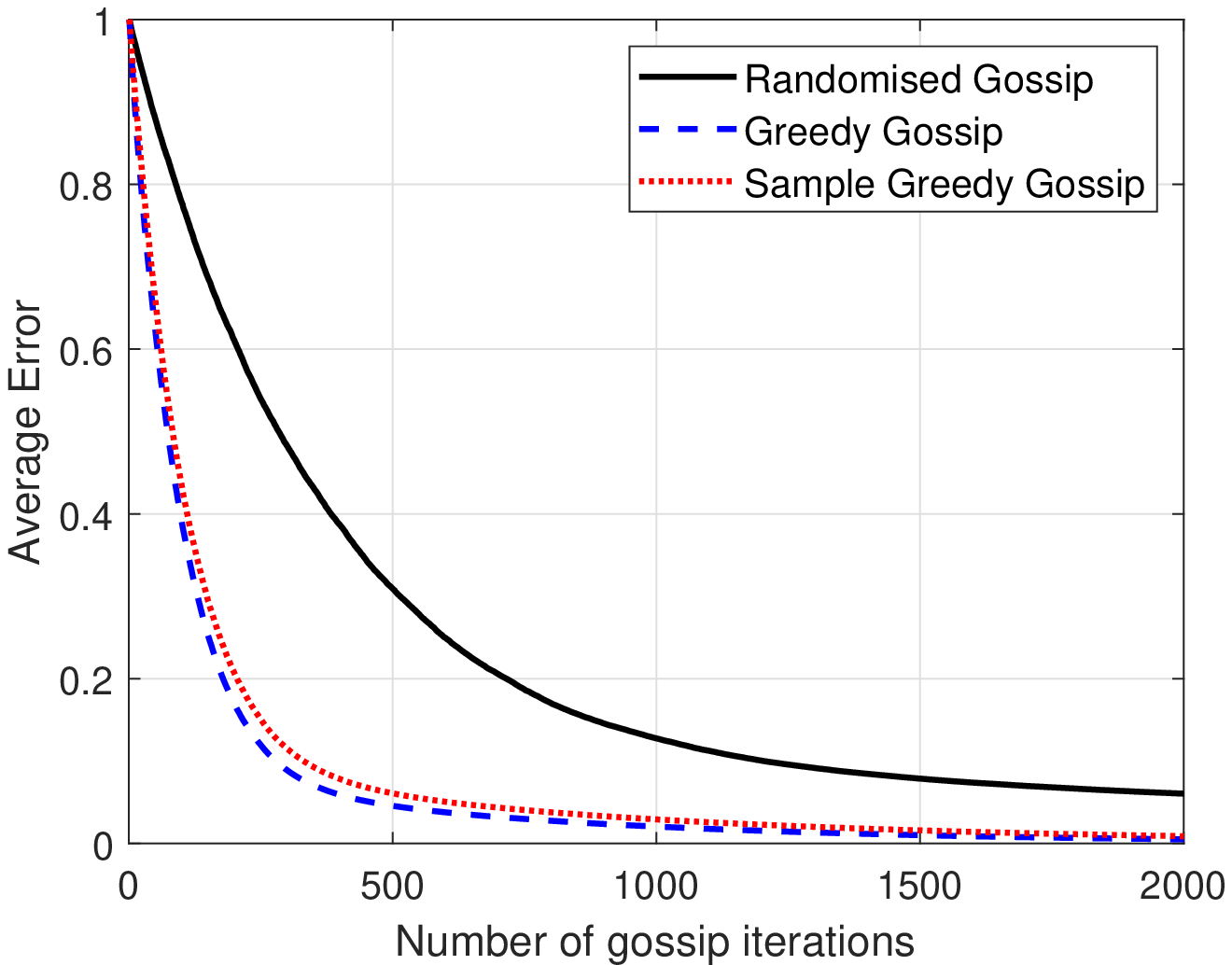}}
	&\subfloat[Linearly-varying field]{\includegraphics[width=.5\linewidth]{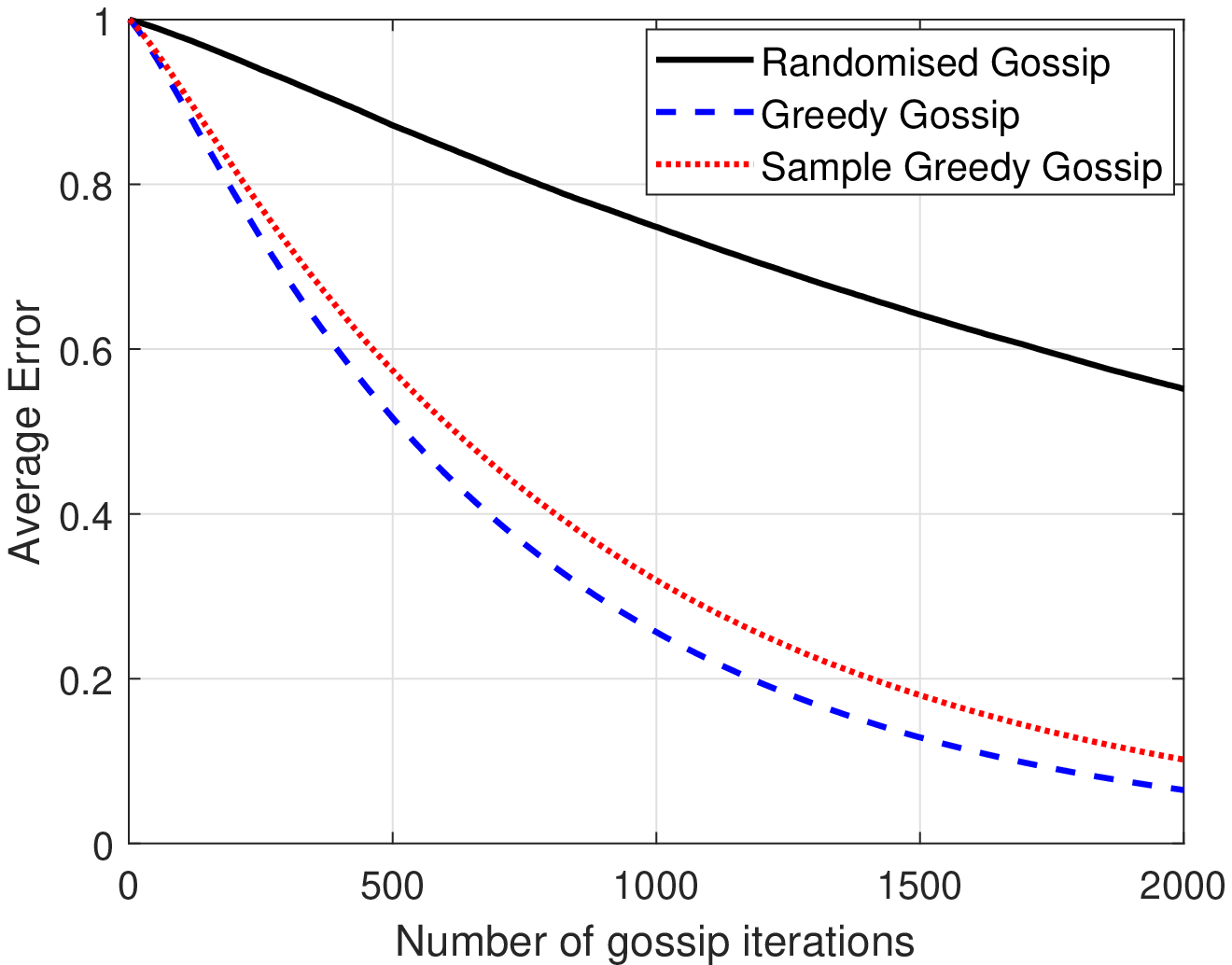}}
\end{tabular}
    \caption{Monte-Carlo comparison results of average convergence error with respect to different number of gossip iterations.}
\label{fig:iteration}
\end{figure}

\subsection{Effect of Node Activation Probability}

Now, let us investigate the effect of the node activation probability on the tradeoff performance of the proposed SGG algorithm. For this purpose, the number of gossip iterations in each Monte-Carlo run and the scaling factor are set as $L=1000$ and $d=2$, respectively. Fig. \ref{fig:probability} presents the comparison results of average convergence error for different gossip algorithms with different node activation probabilities obtained from Monte-Carlo simulations.  

From Fig. \ref{fig:probability}, it can be observed that the greedy gossip provides the best convergence performance among all the tested algorithms. This can be attributed to the fact that the greedy gossip process finds the optimal local node for average computation. However, as stated before, this achievement requires each local node to communicate with all its connected neighbours at each iteration. As a comparison, the proposed SGG offers great flexibility and well balance between communication cost and convergence performance introduced by the stochastic sampling strategy. With the increase of node activation probability, SGG provides improved convergence performance and converges to that of greedy gossip when $p=1$. If the local node cannot provide enough bandwidth for communication, a relatively small node activation probability can be selected to save the communication cost. When $p=0$, the proposed algorithm becomes identical to randomised gossip. The results confirm that the proposed SGG algorithm is a generalised version of the randomised and greedy gossip algorithms. 

\begin{figure}[h!]
\centering
\begin{tabular}{cc}
	\subfloat[Gaussian bumps field]{\includegraphics[width=.5\linewidth]{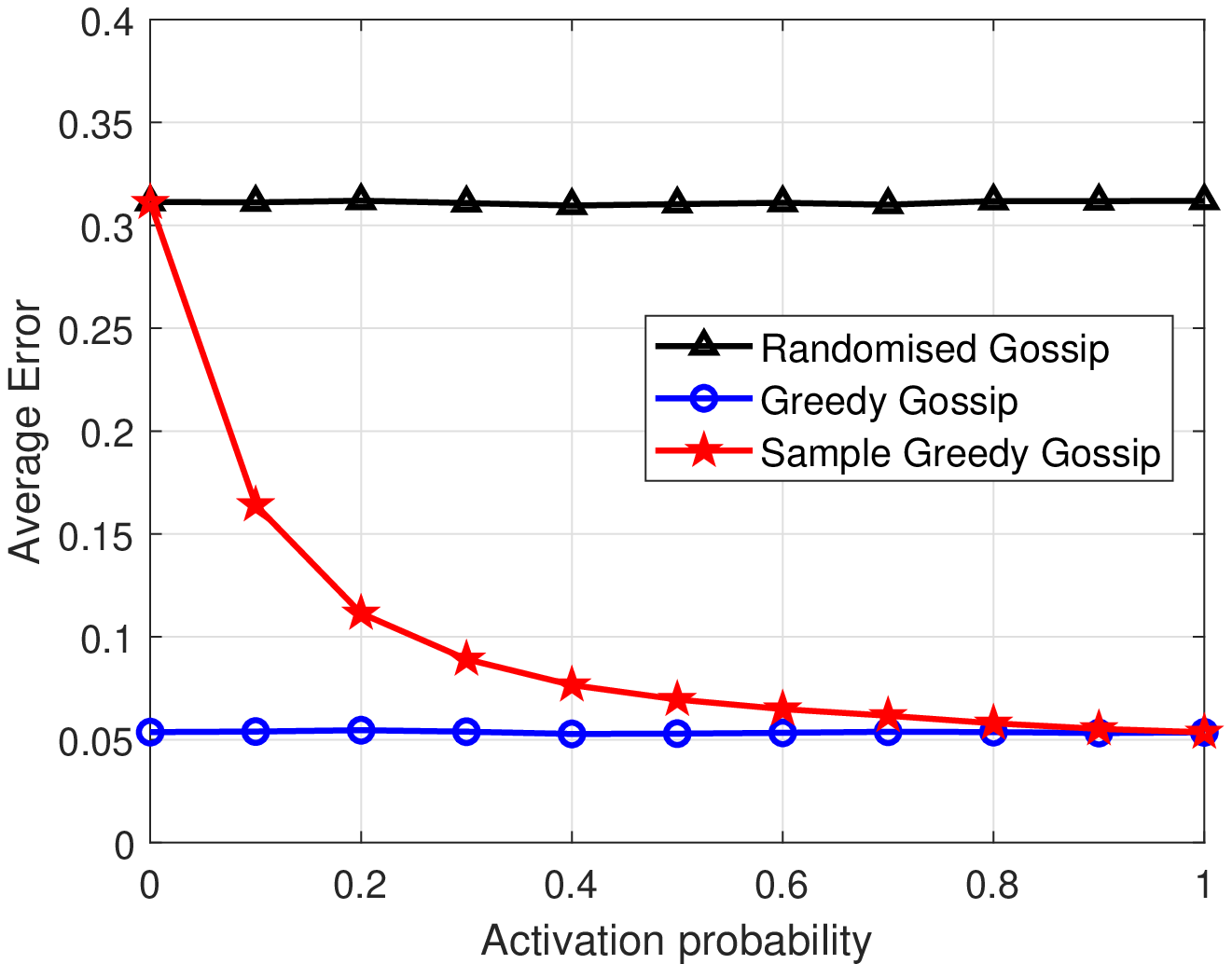}}
	&\subfloat[Spike field]{\includegraphics[width=.5\linewidth]{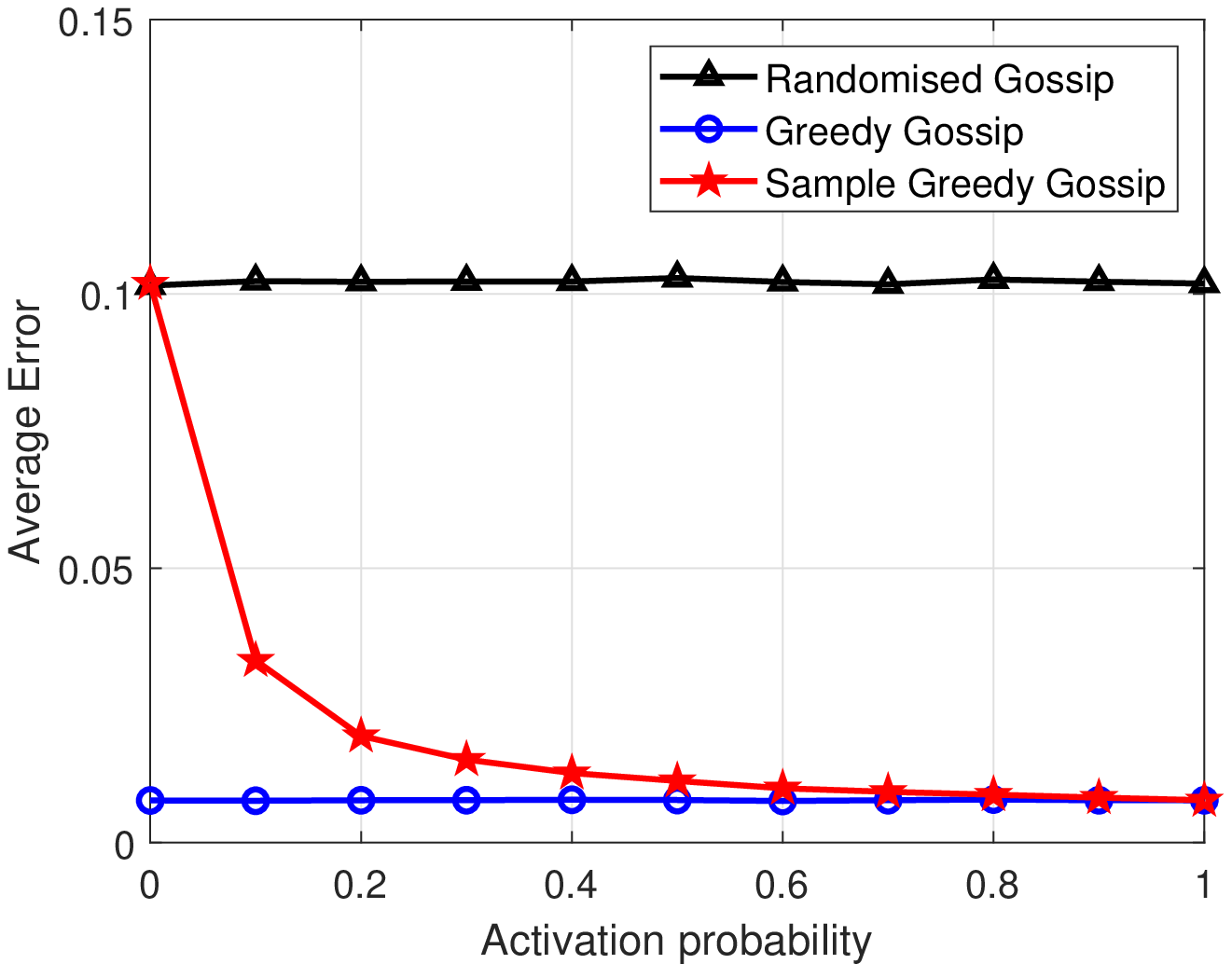}}
\end{tabular}
\begin{tabular}{cc}
	\subfloat[Uniform random field]{\includegraphics[width=.5\linewidth]{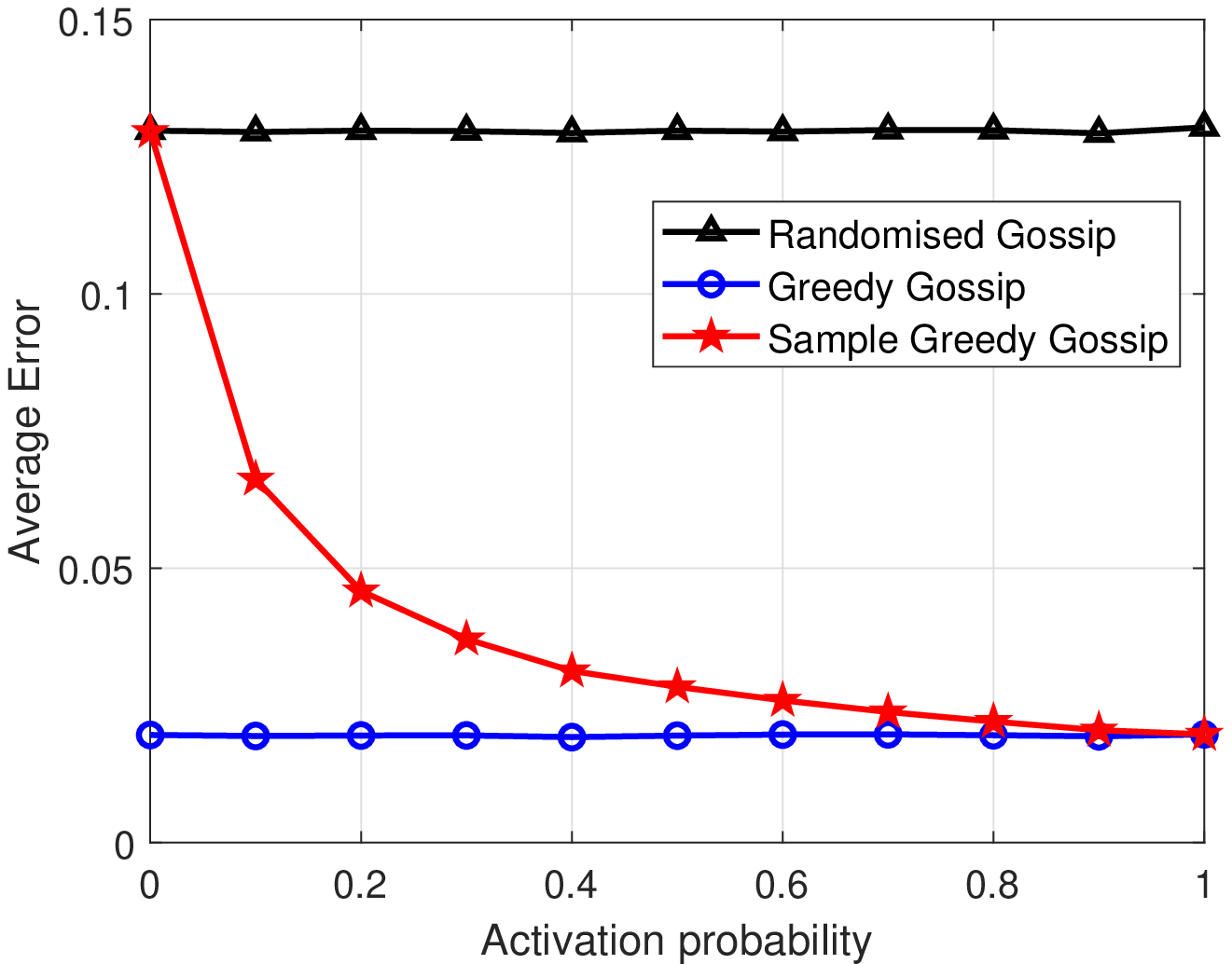}}
	&\subfloat[Linearly-varying field]{\includegraphics[width=.5\linewidth]{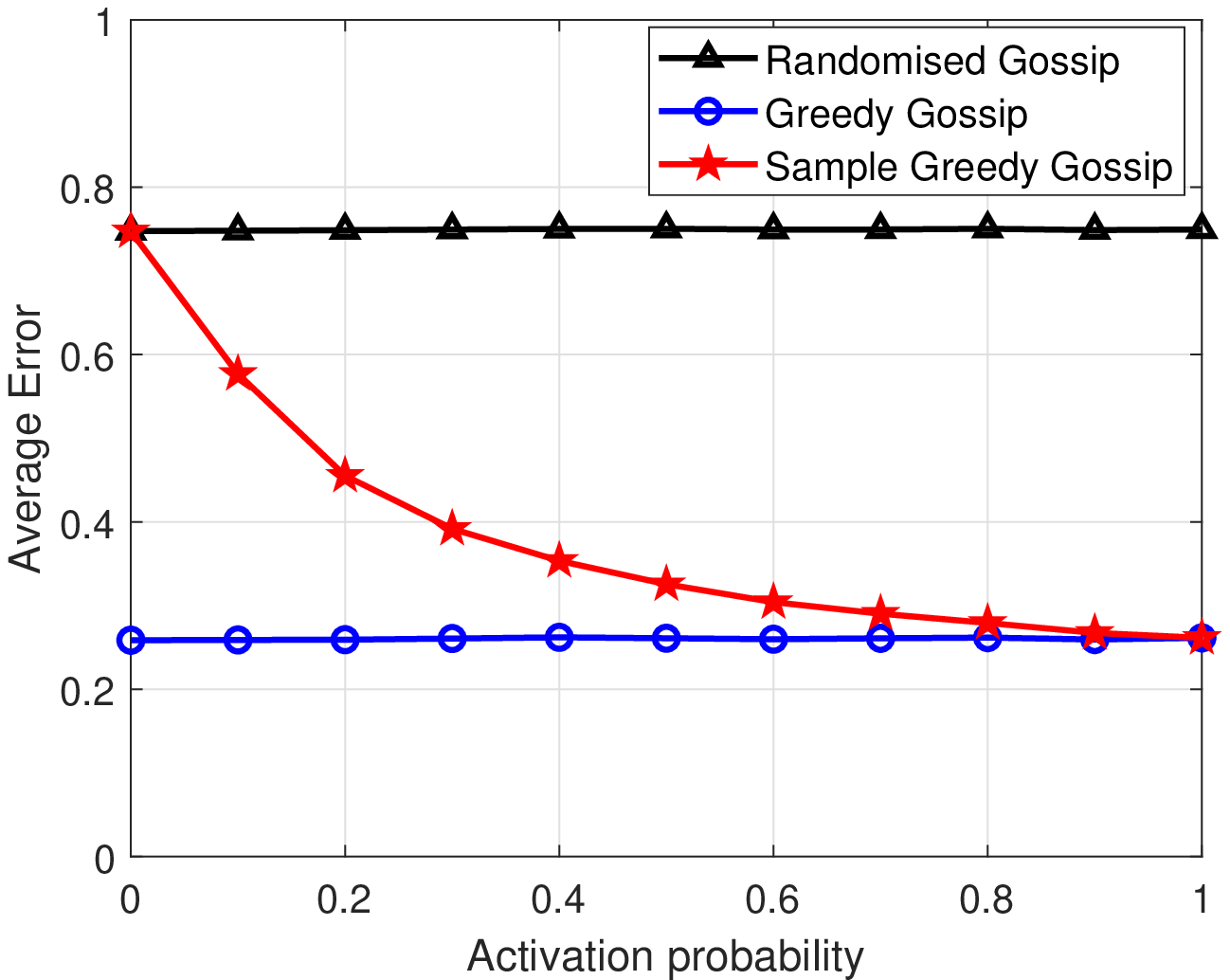}}
\end{tabular}
    \caption{Monte-Carlo comparison results of average convergence error with respect to different node activation probabilities.}
\label{fig:probability}
\end{figure}

\subsection{Effect of Network Sparsity}

Except for the node activation probability, another important factor that affects the performance of the proposed SGG algorithm is the network sparsity or connectivity. For this reason, this subsection empirically analyses the performance of difference gossip algorithms with different network sparsity. Fig. \ref{fig:range} presents the comparison results of average convergence error for different gossip algorithms with different scaling factors obtained from Monte-Carlo simulations. In all simulation runs, the number of gossip iterations in each Monte-Carlo run and the node activation probability are set as $L=1000$ and $p=0.5$, respectively. 

It can be observed from Fig. \ref{fig:range} that the performance discrepancy among all tested gossip algorithms is not much when for sparse network, i.e., smaller $d$. With the increase of the scaling factor $d$, both greedy gossip and the proposed SGG significantly improve the convergence performance. The reason of this fact is clear: the randomised gossip process randomly selects a local node in communication while the other two gossip algorithms either find an optimal or sub-optimal path in information exchange. Therefore, dense networks provide higher possibility to improve the convergence speed with an optimised communication path.

\begin{figure}[h!]
\centering
\begin{tabular}{cc}
	\subfloat[Gaussian bumps field]{\includegraphics[width=.5\linewidth]{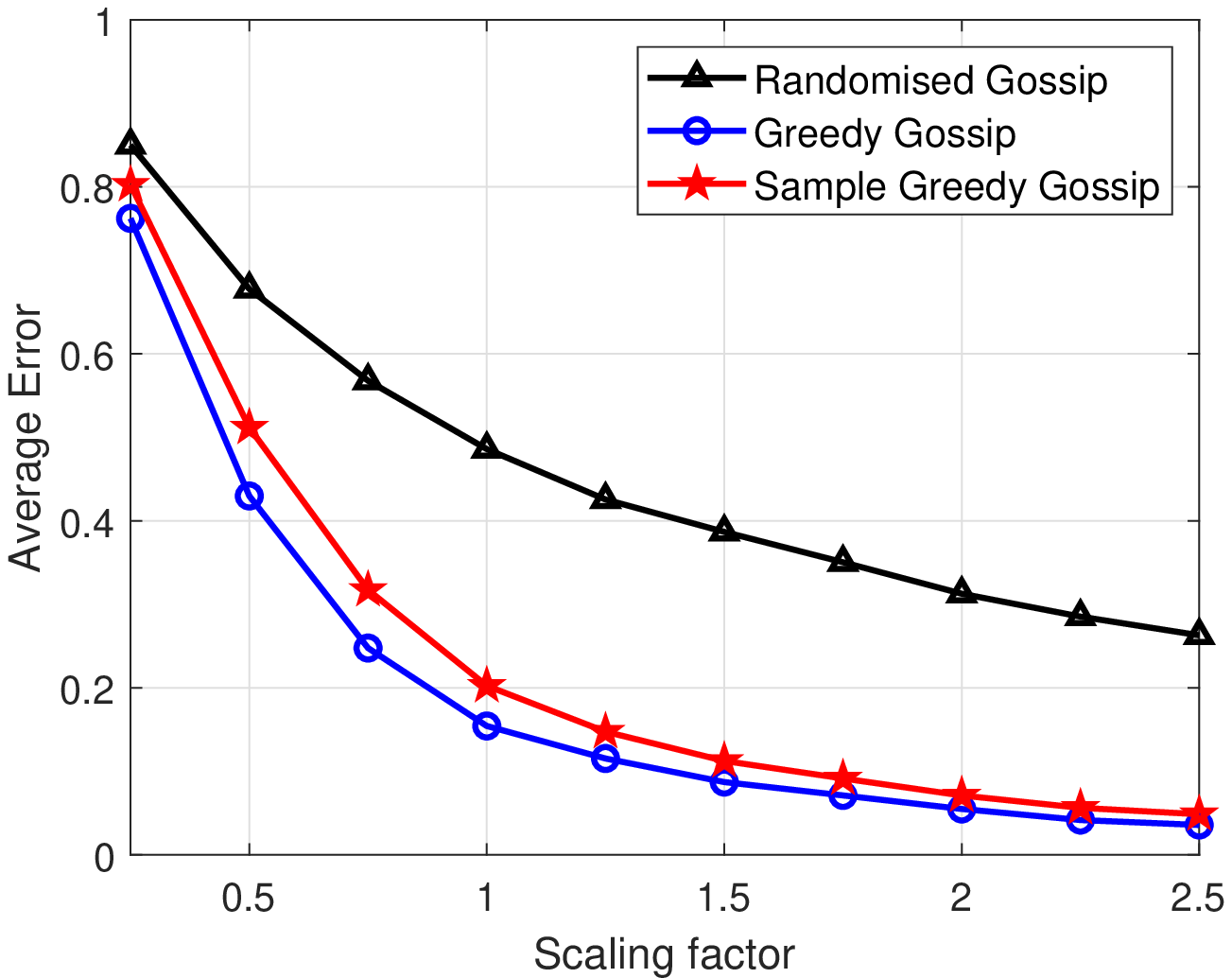}}
	&\subfloat[Spike field]{\includegraphics[width=.5\linewidth]{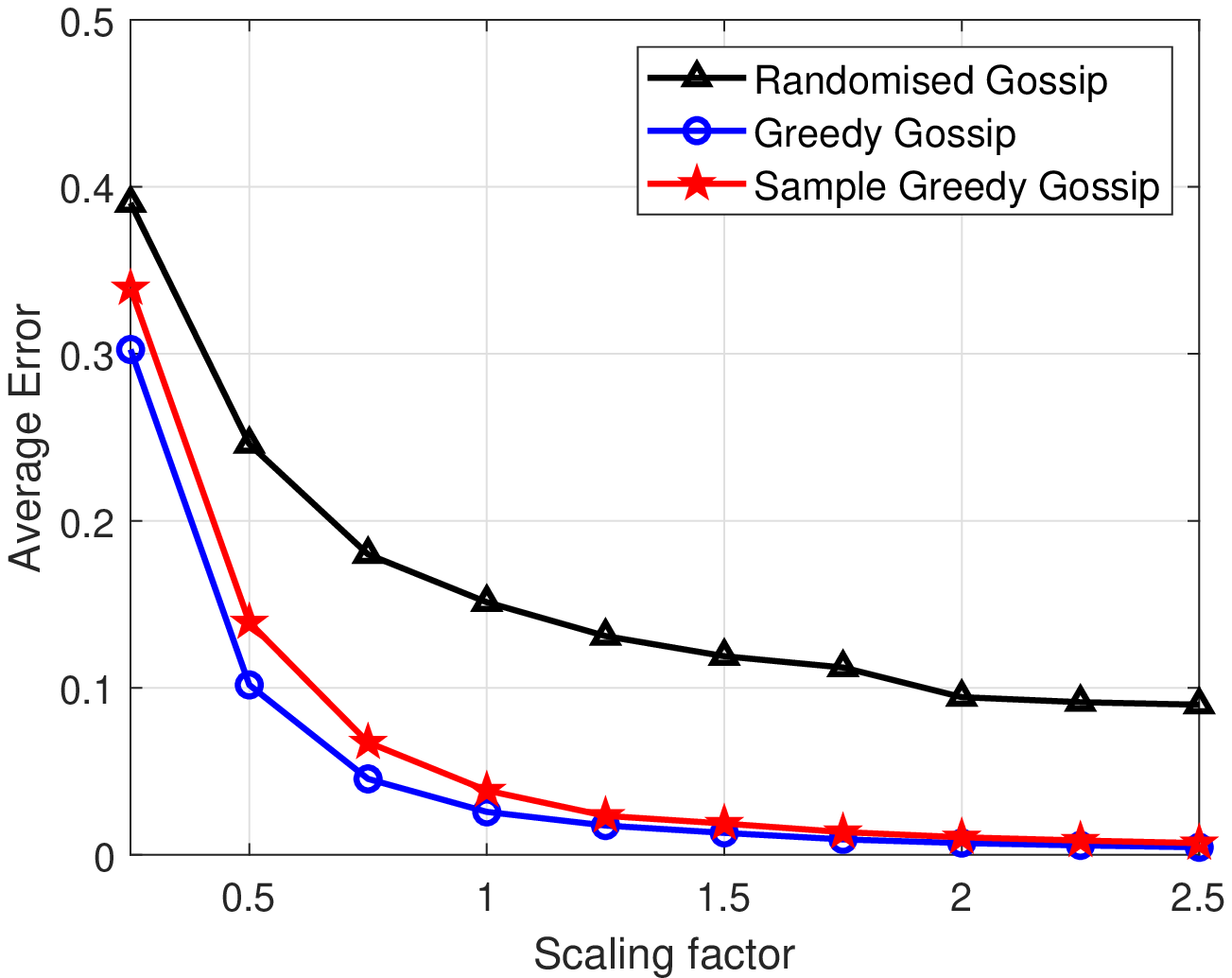}}
\end{tabular}
\begin{tabular}{cc}
	\subfloat[Uniform random field]{\includegraphics[width=.5\linewidth]{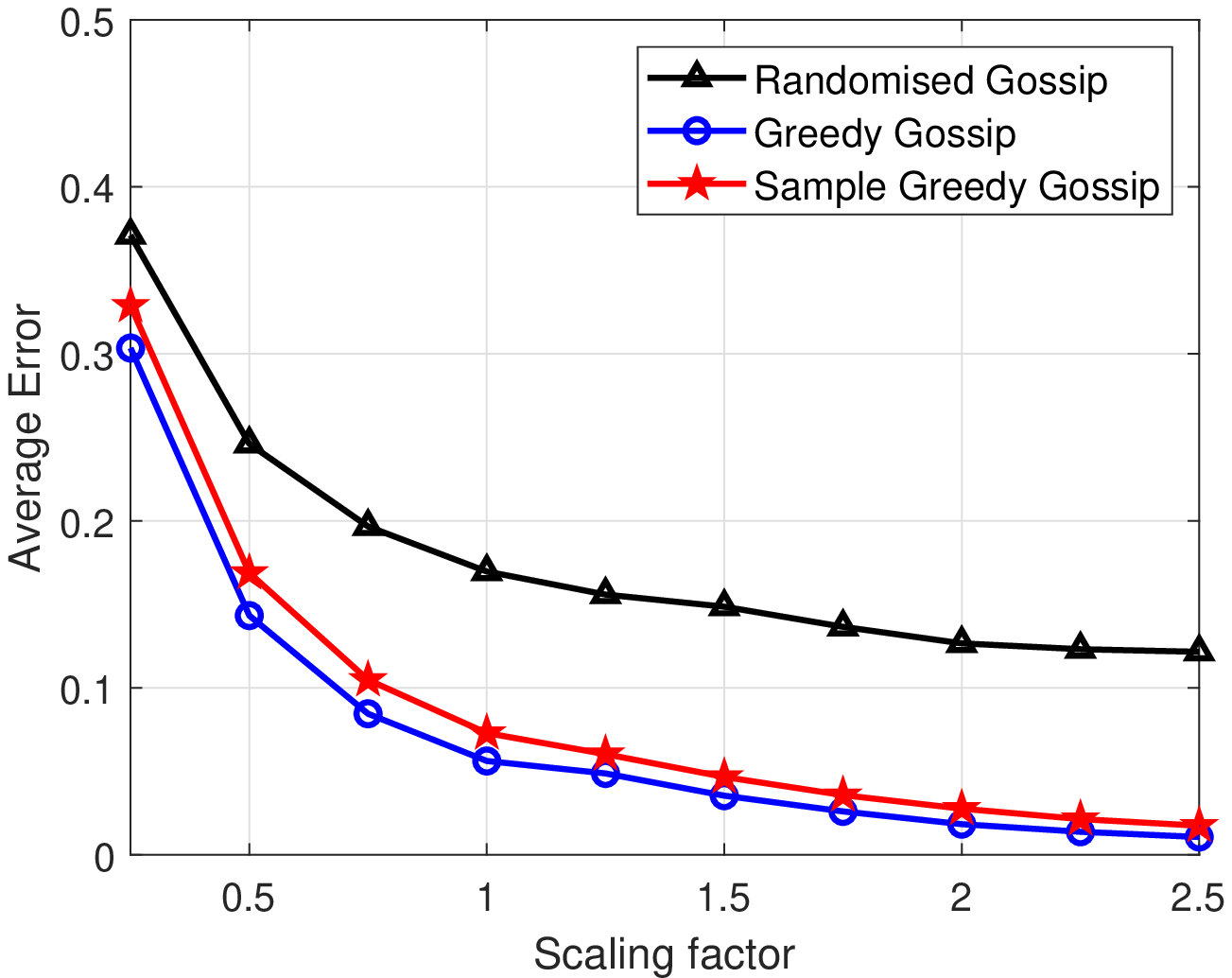}}
	&\subfloat[Linearly-varying field]{\includegraphics[width=.5\linewidth]{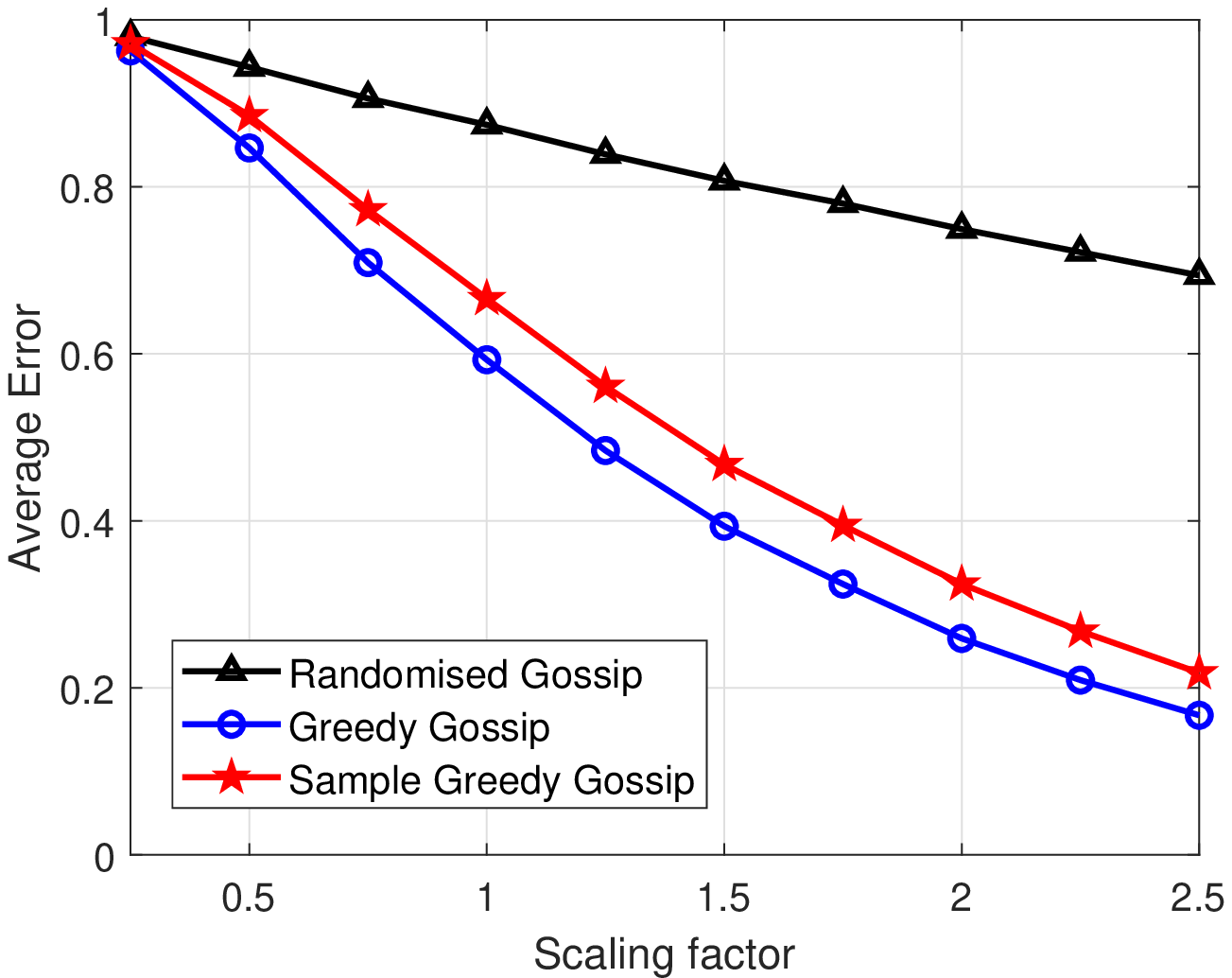}}
\end{tabular}
    \caption{Monte-Carlo comparison results of average convergence error with respect to different network sparsity.}
\label{fig:range}
\end{figure}

\section{Conclusions}
\label{sec:6}

This paper proposed a sample greedy gossip algorithm for average computation over a partially connected network. Rigorous convergence analysis of the proposed sample greedy gossip algorithm is carried out to support its applications. The empirical investigation demonstrates the validity of the theoretical analysis results. Also, theoretical and numerical analysis indicates that the proposed algorithm can be viewed as a generalised version of the randomised and greedy gossip algorithms. The prominent feature of the proposed algorithm lies in that it allows tradeoff  between convergence speed and communication burden: our algorithm shows faster convergence speed, compared to the randomised gossip, and requires less communication burden, compared to the greedy gossip.

\bibliographystyle{IEEEtran}
\bibliography{SGG}

\end{document}